\documentclass[journal, onecolumn]{IEEEtran}

\IEEEoverridecommandlockouts

\usepackage{comment, color, graphicx, enumerate, bm, mathtools, amssymb, subcaption, tikz, tabulary, url, algorithm, listings, algorithmicx, bbm, nth, braket}
\usepackage{hyperref}
\usepackage{cleveref}

\usepackage{float}
\usepackage[noend]{algpseudocode}
\usepackage[font=small,skip=1pt]{caption}

\crefformat{section}{\S#2#1#3} % see manual of cleveref, section 8.2.1
\crefformat{subsection}{\S#2#1#3}
\crefformat{subsubsection}{\S#2#1#3}

\newtheorem{lemma}{Lemma}
\newtheorem{theorem}{Theorem}
\newtheorem{corollary}{Corollary}
\newtheorem{proposition}{Proposition}
\newtheorem{conjecture}{Conjecture}
\newtheorem{definition}{Definition}
\newtheorem{observation}{Observation}
\newenvironment{proof}
{ \begin{IEEEproof} }
{ \end{IEEEproof} }

\newcommand{\norm}[1]{\left\lVert#1\right\rVert}

\allowdisplaybreaks
\makeatletter
\newcounter{longaligned}
\newenvironment{longaligned}[1][]
 {%
  \stepcounter{longaligned}%
  \refstepcounter{equation}%
  \label{longaligned@\thelongaligned}%
  #1%
  \start@align\@ne\st@rredtrue\m@ne % start align*
 }
 {\endalign}

\begin{document}
\title{Simplex Queues for Hot-Data Download}

\author{Mehmet Fatih Akta\c{s},
Elie Najm,
and Emina Soljanin,~\IEEEmembership{Fellow,~IEEE}%
\thanks{M. Aktas and E. Soljanin are with Rutgers University. E-mail: \{mehmet.aktas, emina.soljanin\}@rutgers.edu.}%
\thanks{E. Najm is with EPFL. E-mail: elie.najm@epfl.ch.}%
\thanks{This paper was presented in part at 2017 ACM SIGMETRICS/International Conference on Measurement and Modeling of Computer Systems.}}

\maketitle

% #######################################  Abstract  ####################################### #
\begin{abstract}
% Why
In cloud storage systems, hot data is usually replicated over multiple nodes in order to accommodate simultaneous access by multiple users as well as increase the fault tolerance of the system. Recent cloud storage research has proposed using availability codes, which is a special class of erasure codes, as a more storage efficient way to store hot data. These codes enable data recovery from multiple, small disjoint groups of servers. The number of the recovery groups is referred to as the \emph{availability} and the size of each group as the \emph{locality} of the code. Until now, we have very limited knowledge on how code locality and availability affect data access time. Data download from these systems involves multiple fork-join queues operating in-parallel, making the analysis of access time a very challenging problem.
% What
In this paper, we present an approximate analysis of data access time in storage systems that employ simplex codes, which are an important and in certain sense optimal class of availability codes.
We consider and compare three strategies in assigning download requests to servers; first one aggressively exploits the storage availability for faster download, second one implements only load balancing, and the last one employs storage availability only for hot data download without incurring any negative impact on the cold data download.
%These strategies are compared in terms of the gain and pain on the download time of hot and cold data.

% Using a queueing theoretic approach, we derive bounds and approximations on the average data access time under three different request scheduling strategies; first one implements aggressive exploitation of redundancy, second one implements only load balancing with no exploitation of redundancy, and the last one aims to implement fairness by exploiting redundancy only for hot data with zero effect on the download of cold data. We compare these strategies in terms of the gain and pain they incur on the download time of hot and cold data.
\end{abstract}

% ##########################################  Introduction  ######################################### #
\section{Introduction}
In distributed systems, reliable data storage is accomplished through redundancy, which has traditionally been achieved by simple replication of data across multiple nodes \cite{GFS:GhemawatGL03, borthakur2007hadoop}. A special class of erasure codes, known as locally repairable codes (LRCs) \cite{gopalan2012locality}, has started to replace replication in practice \cite{huang2012erasure, sathiamoorthy2013xoring} as a more storage efficient way to provide a desired reliability.

A storage code has locality $r$ and availability $t$ if each data symbol can be recovered from $t$ disjoint groups of at most $r$ servers. Low code locality is desired to limit the number of nodes accessed while recovering from a node failure. Code availability provides a means to cope with node failures and skews in content popularity as follows. First, when data is available in multiple recovery groups, simultaneous node failures have lower chance of preventing the user from accessing the desired content.
Second, frequently requested \textit{hot data} can be simultaneously served by the node at which the data resides as well as multiple groups of servers that store the recovery groups the desired data. Popularity of the stored content in distributed systems is shown to exhibit significant variability. Data collected from a large Microsoft Bing cluster shows that 90\% of the stored content is not accessed by more than one task simultaneously while the remaining 10\% is observed to be frequently and simultaneously accessed \cite{CopingWithSkewedContentPopularityInMapreduce:AnanthanarayananAK11}.
LRCs with good locality and availability have been explored and several construction methods are presented in e.g., \cite{tamo2014family, rawat2014locality}.
% 26\% of the stored content is accessed no more than once over the entire duration of the content, and

It has recently been recognized that the storage redundancy can also provide fast data access \cite{joshi2012coding}. Idea is to simultaneously request data from both the original and the redundant storage, and wait for the fastest subset of the initiated downloads that are sufficient to reconstruct the desired content. Download with redundant requests is shown to help eliminating long queueing and/or service times (see e.g. \cite{joshi2012coding, Codes&Qs:HuangPZ12, gardner2015reducing, MeanFieldAnalysisCodingVsRep:LiRS16, LatencyAnalysisMDSorRep:Parag17, Rlog:JoshiSW17, DecouplingSlowdownJobsize:GardnerHS17, MDSn2:AktasS18} and references therein).
Most of these papers consider download of all jointly encoded pieces of data, i.e., the entire file. Very few papers have addressed download in low traffic regime of only some, possibly hot, pieces of data that are jointly encoded with those of less interest \cite{kadhe2015analyzing, kadhe2015availability, HD:ShuaiL16}.

In this paper, we are concerned with hot data download from systems that employ an LRC with availability for reliable storage \cite{tamo2014family, rawat2014locality}.
In particular, we consider simplex codes, which is a subclass of LRCs with availability. Three reasons for this choice are given as follows:
1) Simplex codes are optimal in several ways, e.g., they meet the upper bound on the distance of LRCs with a given locality \cite{BoundsOnSizeOfLRCs:CadambeM15}, they are shown to achieve the maximum rate among the binary linear codes with a given availability and locality two \cite{RateOptimalityOfSimplex:KadheC17}, they meet the Griesmer bound and are therefore linear codes with the lowest possible length given the code distance \cite{Klein2004griesmer},
2) Simplex codes are minimally different from replication, in that when data is simplex coded any single node failure can be recovered by accessing two nodes, while accessing a single node is sufficient to recover replicated data,
3) Simplex codes have recently been shown to achieve the maximum fair robustness for a given storage budget against the skews in content popularity \cite{ServiceCapacity:AktasAJ17}.

In a distributed system that employs a simplex code, hot data can be downloaded either from a single {\it systematic} node that stores the data or from any one of the $t$ pairs of nodes (i.e., {\it recovery groups}) where the desired data is encoded jointly with others, or redundantly by some subset of these options (see Fig.~\ref{fig:fig_access_schemes}).
We consider three strategies for scheduling the download requests:
1) \emph{Replicate-to-all} where each arriving request is simultaneously assigned to its systematic node and all its recovery groups,
2) \emph{Select-one} where each arriving request is assigned either to its systematic node or to one of its recovery groups,
3) \emph{Fairness-first} where each arriving request is primarily assigned to its systematic node while only hot data requests are opportunistically replicated to their recovery groups when they are idle.
The first two scheduling strategies are the two polarities between download with redundant requests and plain load balancing, while the third aims to exploit download with redundancy while incurring no negative impact on the cold data download time.

\begin{figure}[t]
  \centering
  \begin{tikzpicture}
    \node at (0,0) {\includegraphics[scale=1.2]{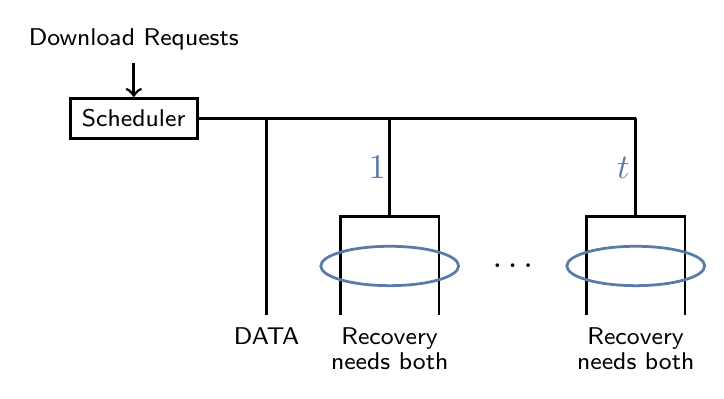}};
  \end{tikzpicture}
  \vspace*{-0.15cm}
  \caption{Data access model in a system that employs a simplex code with availability $t$. Desired data can be accessed by downloading either from the systematic node that stores the data or any of the $t$ recovery groups of two nodes or redundantly by some subset of these options. Download from a recovery group requires fetching the coded content from both servers and recovering the desired data.}
%   Under {\it replicate-to-all,} each arriving request is dispatched to the systematic node containing the desired data and simultaneously to all its $t$ recovery groups.
%   Under {\it select-one,} each arriving request is assigned randomly either to the systematic node or one of its recovery groups. One can consider other access schemes between these two polarities.
  \label{fig:fig_access_schemes}
\end{figure}

A download time analysis for storage systems that employ LRCs is given in \cite{kadhe2015analyzing, kadhe2015availability} by assuming a low traffic regime, in that no request arrives at the servers before the request in service departs, so the adopted system model does not involve any queues or any other type of resource sharing. When the low traffic assumption does not hold, a system that employs an LRC with availability consists of multiple inter-dependent fork-join queues. This renders the download time analysis very challenging even under Markovian arrival and service time models.

In this paper, we present a first attempt on analyzing the download time with no low traffic assumption in systems that employ LRCs with availability, simplex codes in particular, for reliable storage.
Under replicate-to-all scheduling, system implements multiple inter-dependent fork-join queues, hence the exact analysis is formidable due to the infamous state explosion. We outline a set of observations and heuristics that allow us to approximate the system as an $M/G/1$ queue. Using simulations, we show that the presented approximation gives an accurate prediction of the average hot data download time at all arrival regimes.
Under select-one scheduling, the results available in the literature for fork-join queues with two servers can be translated and used to get accurate approximations of the average download time.
At last, we study download time under fairness-first scheduling, and evaluate the pain and gain of replicate-to-all (aggressive exploitation of redundancy) compared to fairness-first (opportunistic exploitation of redundancy).

The remainder of the paper is organized as follows. In Sec.~\ref{sec:sec_sys_model} we define a simplex coded storage model, and explain the queueing and service model that is adopted to study the data access time. Sec.~\ref{sec:sec_reptoall} introduces the replicate-to-all scheduling, explains the difficulties it poses for analyzing data access time and presents an approximate method for analysis. Sec.~\ref{sec:sec_reptoall_t1} presents the analysis for systems with availability one and Sec.~\ref{sec:sec_reptoall_t_conj_approx} generalizes the analysis for systems with any degree of availability.
In Sec.~\ref{sec:sec_selectone} we present the select-one scheduling and compare its data access time performance with replicate-to-all scheduling. Sec.~\ref{sec:sec_fairnessfirst} introduces the fairness-first scheduler and shows the gain and pain of replicate-to-all over fairness-first in hot and cold data access time.

% \textbf{Contribution:}
% We consider two request arrival scenarios. In \emph{fixed-arrival} scenario, requests that arrive in a busy period are for only one data symbol. In \emph{mixed-arrival} scenario, requests may arrive for downloading any data symbol during a busy period.
% We give an analysis for the system under \emph{replicate-to-all} scheduling strategy.
% Starting with the simplest possible simplex code with availability $1$, we obtain close lower and upper bounds on the average access time to hot-data. Then we argue that extending these results further to codes with higher availability is hard. However, using ideas from queueing and renewal theory, we obtain lower and upper bounds, and an approximation for the average access time to hot-data. Finally, \emph{select-one} scheduling is compared with the replicate-to-all scheduling in terms of access time.

% #########################################  System Model  ######################################## %
\section{System Model}
\label{sec:sec_sys_model}
A binary simplex code is a linear systematic code. With an $[n=2^k-1, k]$ simplex code, a data set $[d_1, \ldots, d_k]$ is expanded into $[d_1, \ldots, d_k, c_1, \ldots c_{n-k}]$ where $c_i$'s represent the added redundant data. Each data symbol in the resulting expanded data set is stored on a separate node.
For a given $k$ and $n$, simplex codes have optimal distance and a simplex coded storage can recover from simultaneous failure of any $2^{k-1} - 1$ nodes. Availability of a simplex code is also given as $t = 2^{k-1} - 1$ and its locality is equal to two, that is, any encoded data $d_i$ is available in a systematic node or can be recovered from $t$ other disjoint pairs of nodes. For instance, data set $[a, b, c]$ would be expanded into $[a, b, a+b, c, a+c, b+c, a+b+c]$ in a $[7, 3]$ simplex coded storage.

% When hot data symbol $a$ is encoded together with the cold data symbols $b$ and $c$ using a $[7,3]$ simplex code, each symbol in the expanded data set $[a, b, c, a+b, a+c, b+c, a+b+c]$ is stored on a separate node. In this example, hot data $a$ is available in its systematic node and within three other recovery groups $(b, a+b)$, $(c, a+c)$, $(b+c, a+b+c)$. In general, when a hot data symbol is encoded together with $log_2(t+1)$ cold data symbols, hot data is available in its systematic server and within $t$ disjoint recovery groups.

In the data download scenario we consider here, download requests arrive for the individual data symbol $d_i$'s rather than the complete data set.
We assume that one of the jointly encoded data symbols is more popular than the others at any time. The popular data symbol is referred as the hot data while the others are referred as cold data.
Download requests are dispatched to the storage servers at their arrival time to the system. We consider three request scheduling strategies.
The first one is {\it replicate-to-all} scheduling which aggressively implements download with redundancy for all arriving requests. Each request is assigned to its systematic server that stores the desired data, and simultaneously replicated to all its $t$ recovery groups at which the requested data can be reconstructed. A request is completed and departs the system as soon as the fastest one of its $t+1$ dispatched copies completes, and its remaining $t$ outstanding copies get immediately removed from the system.
The second one we consider is {\it select-one} scheduling which does not employ redundancy for download and simply implements load-balancing, in that each arriving request is assigned either to its systematic server or one of its recovery groups.
In the download time analysis of these two scheduling strategies, download traffic for cold data is assumed to be negligible and the goal of the system is to finish the arriving hot data requests as fast as possible. We compare the performance of these two polarities of scheduling in terms of the system stability and the average hot data download time.

At last, we consider the scenario in which cold data arrival traffic is non-negligible while the cold data requests are assumed to arrive at a slower rate than the hot data requests. Hot data download with redundancy can reduce the download time and allow the hot data requests to leave the system faster. However, redundant copies of the hot data requests occupy the cold data servers in the recovery groups, which causes additional waiting time for the cold data requests and may lead to dramatically higher cold data download times. This cannot maintain fairness across the hot and cold data requests, which is an important goal of scheduling in systems with redundancy \cite{SchingRedForFairness:GardnerHH17}.
We introduce and study {\it fairness-first} scheduling which exploits redundancy for hot data download opportunistically with no effect on the cold data download. Each arriving request is primarily assigned to its systematic server, and only the hot data requests are replicated to the recovery groups. A redundant copy of a hot data request  is accepted into service only if the cold data server in the assigned recovery group is idle. Even if a redundant hot data request copy is let to start service at a recovery group, it is immediately removed from service as soon as a cold data request arrives to the occupied cold data server.

We assume for tractable analysis that download requests arrive to system according to a Poisson process of rate $\lambda$.
Each server can serve only one request at a time and the requests assigned to a server wait for their turn in a FIFO queue. We refer to the random amount of time to fetch the data at a server and stream it to the user as the \textit{server service time}.
Download of hot data at a recovery group requires fetching content from both servers. Thus, redundant hot data requests assigned to a recovery group are {\it forked} into the two recovery servers, which then need to {\it join} to complete a download.

Initial or redundant data stored at the storage nodes have the same size. We assume independent service time variability for content download at each server. This allows us to model the service times with an identical and independent (i.i.d.) distribution across the requests and the servers.
We firstly study the system by assuming exponential service times for analytic tractability, then extend the analysis in some cases to general service times. Note that, exponential service times cannot model the effect of data size on the download time since a sample from it can be as small as zero. Therefore, it is an appropriate model only when the effect of data size on the service times is negligible \cite{DecouplingSlowdownJobsize:GardnerHS17}.

% We are interested in the time to access hot data in a simplex coded system. We firstly consider the case where request arrivals in a busy period only ask for the hot data symbol, in other words, when the request traffic for the cold data symbols is negligibly small (Sec.~\ref{sec:sec_reptoall} and \ref{sec:sec_selectone}). Later, we study the case when requests arrive also for cold data download (Sec.~\ref{sec:sec_fairnessfirst}).

% #########################################  Replicate-to-all  ######################################## %
\section{Replicate-to-all Scheduling}
\label{sec:sec_reptoall}
In this section, we consider hot data download in a binary simplex coded storage under replicate-to-all scheduling.  Recall that in a Microsoft Bing data center, 90\% of the stored content is reported to be not accessed frequently or simultaneously \cite{CopingWithSkewedContentPopularityInMapreduce:AnanthanarayananAK11}. Adopting this observation in our analysis of the hot data access time under replicate-to-all scheduling, download traffic for cold data is assumed to be negligible. Subsec.~\ref{subsec:subsec_reptoall_t_mixed_arrival} relates the analysis to the scenario where cold data download traffic is non-negligible.
Later in Sec.~\ref{sec:sec_fairnessfirst}, we study the scenario with non-negligible cold data download traffic and discuss the impact of replicate-to-all scheduling on hot and cold data download time.

Replicate-to-all scheduler replicates each arriving hot data request to its systematic server as well as to all its $t$ recovery groups (see Fig.~\ref{fig:fig_access_schemes}).
A request is completed as soon as one of its copies completes at either the systematic server or one of the recovery groups. As soon as a request is completed, all its outstanding copies get removed from the system immediately.
Download from each recovery group implements a separate fork join queue with two-servers.

Analysis of fork-join queues is a notoriously hard problem and exact average response time is known only for a two-server fork join queue with exponential service times \cite{flatto1984two}.
Moreover, fork-join queues in the replicate-to-all system are inter-dependent since the completion of a request copy at the systematic server or at one of the fork-join queues triggers the cancellation of its outstanding copies that are either waiting or in service at other servers.
Given this inter-dependence, exact analysis of the download time in replicate-to-all system is a formidable task.

Derivation of the average hot data download time is presented in \cite{kadhe2015analyzing} under the low traffic assumption, that is, assuming no queues or resource sharing in the system. When low traffic assumption does not hold and queueing analysis is required, an upper bound on the download time is found in \cite{kadhe2015analyzing} by using the more restrictive split-merge service model. Split-merge model assumes that servers are not allowed to proceed with the requests in their queues until the request at the head of the system departs. Simulation results show that the upper bound obtained by the split-merge model is loose unless the arrival rate is very low.
We here present a framework to approximate the replicate-to-all system as an $M/G/1$ queue under any arrival rate regime, which proves to be an accurate estimator of the average hot data download time. The presented approximation provides a fairly simple interpretation of the system dynamics and employs heuristics that shed light on the performance of aggressive exploitation of storage availability.

% -----------------------  M/G/1 Approximation  --------------------------- %
\subsection{$M/G/1$ Approximation}
\label{subsec:subsec_reptoall_mg1_approx}
% We here state the observations that lead us to approximate the replicate-to-all system as an $M/G/1$ queue.
Each arriving request has one copy at the systematic server and $t$ redundant copies split across the recovery groups. A request departs from the system as soon as one of its copies completes.
Request copies at the systematic server wait and get served in a FIFO queue, hence departures from the systematic server can only be in the order of arrival.
Request copies that are assigned to recovery groups are forked into two siblings and both have to finish service for the request completion.
One of the recovery servers can be ahead of the other one in the same group in executing the requests in its queue. We refer to such servers as a {\it leading server}.
Although one of the forked copies may happen to be served by a leading server and finish service earlier than its sibling, the remaining copy has to wait for the other request copies in front of it. Thus, departures at the recovery groups can also be only in the order of arrival.
These imply that a request cannot exit the replicate-to-all system before the other requests in front of it, that is, requests depart the system in the order they arrive.

% Within the systematic server and each of the fork-join queues, requests depart in the order they arrive.
% Given also that all the redundant copies of a request are removed from the system as soon as the fastest replica completes, it is easy to see that requests depart the replicate-to-all system in the order they arrive.

Exact analysis of download time requires keeping track of each copy of the requests in the system.
There can be multiple (at most $t$) leading recovery servers in the system simultaneously and each can be ahead of its slower peer by at most as many requests as in the system. These cause the infamous state explosion problem and renders the exact analysis intractable.
Due to the leading servers, copies of a single request can start service at different times and multiple requests can be in service (there can be at most $t+1$ different requests in service simultaneously), which complicates the analysis even further.
In the following, we redefine the service start time of the requests, which allows us to eliminate this last mentioned complication.
\begin{definition}
  Service start time of a request is defined as the first time epoch at which all its copies are in service, i.e., when none of its copies is left waiting in queue.
\label{def_req_servstart}
\end{definition}
Given this new definition of request service start times, there can be only one request in service while the system is busy.
For two requests to be simultaneously in service, all copies of each must be in service simultaneously, which is impossible given that requests depart the system in order.
This observation is crucial for the rest of the section and stated below to be able to refer to it later.
\begin{observation}
  Requests depart the replicate-to-all system in the order they arrive and there can be at most one request in service at any time given Def.~\ref{def_req_servstart}.
\label{obv_reptoall}
\end{observation}

Once a request starts service according to Def.~\ref{def_req_servstart}, the layout of its copies at the servers determines its service time distribution. For instance, if all copies of a request remain in the system at its service start time epoch, then the request will be served according to
\[ S_0 \sim \min\{\tilde{V}, \tilde{V}^1_{2:1}, \dots, \tilde{V}^t_{2:1}\} \]
where $\tilde{V}$ denotes the residual service time of the copy at the systematic server, and $\tilde{V}^i_{2:1}$'s denote the maximum of the two residual service times of the sibling copies at each recovery group.

A request copy can depart from a leading server before the request starts service. When $j$ copies of a request depart earlier, we denote the service start layout of the request as type-$j$, for which the service time distribution is given as
\[ S_j \sim \min\{\tilde{V}_0, \dots, \tilde{V}_j, \tilde{V}^1_{2:1}, \dots, \tilde{V}^{t-j}_{2:1}\} \]
for $j= 0, 1, \dots, t$. $\tilde{V}_0$ is the residual service time at the systematic server, $\tilde{V}_1, \dots, \tilde{V}_j$ denote the residual service times at the $j$ recovery groups with leading servers and $\tilde{V}^i_{2:1}$'s denote the maximum of the two residual service times at the remaining $t-j$ recovery groups without a leading server.
Notice that the previous example with no early departure of any request copy corresponds to type-$0$ service start.

The layout of the copies of a request at the servers, hence its service time distribution, is determined by the number of copies that depart before the request starts service. These early departures can only be from the leading servers so the service time distribution of a request depends on the difference between the queue lengths at the leading servers and their slow peers. In other words, service time distribution of a request is dictated by the system state at its service start time epoch.
Queue lengths carry memory between service starts of the subsequent requests. For instance, if a request makes a type-$j$ service start with all the leading servers being ahead of their slow peers by at least two, then the service start type of the subsequent request will be at least $j$. Therefore, in general, service time distributions are not independent across the subsequent requests.

When service times at the servers are not exponentially distributed, there are infinitely many possible distributions for the request service times. This is because some copies of a request may start earlier and stay in service until the request moves to head of the line and starts service. Residual service time of these early starters is differently distributed than the request copies that move to service at the request service start time epoch.
We can trim the number of possible request service time distributions down to $t+1$ by modifying the system such that each copy of a request remaining in service is restarted at the service start time of the request. This modification is essential for approximating the download time when the service times at the servers are not exponentially distributed, as further discussed in Sec.~\ref{subsec:subsec_reptoall_t_nonexp_servers}.
When service times at the servers are exponentially distributed, this modification is not required. Memoryless property allows us to treat the copies that start service early as if they move to service at the request service start time.
\begin{lemma}
  In replicate-to-all system, if service time $V_i$'s at the servers are exponentially distributed with rate $\mu$, an arbitrary request can be served with any of the $t+1$ different types of distributions. Type-$j$ distribution for $j = 0, \dots, t$ is given as
  \begin{equation}
  \begin{split}
    S_j \sim \min\{V, \ldots, V_j, V^1_{2:1}, \dots, V^{t-j}_{2:1}\}
  \end{split}
  \label{eq:eq_V_j}
  \end{equation}
  where $V^i_{2:1}$'s are distributed as the maximum of the two independent $V_i$'s.
  
  Then, first and second moments of type-$j$ request service time distribution are given as
  \begin{equation}
    \begin{split}
      E[S_j] &= \sum\limits_{k=0}^{t-j} {{t-j}\choose{k}} \frac{2^k (-1)^{t-j-k}}{\mu(2t+1-j-k)}, \\
      E[S_j^2] &= \sum\limits_{k=0}^{t-j} {{t-j}\choose{k}} \frac{2^{k+1} (-1)^{t-j-k}}{\bigl[\mu(2t+1-j-k )\bigr]^2}.
  \end{split}
  \label{eq:eq_serv_dist_moments}
  \end{equation}
  All moments of $S_j$ decrease monotonically in $j$.
\label{lm_reqservtime_dists}
\end{lemma}
\begin{proof}
  \eqref{eq:eq_V_j} follows from the memoryless property of exponential service times at the servers, then by the law of total expectation \eqref{eq:eq_serv_dist_moments} is obtained.
  
  One can easily show that the tail probability of type-$j$ request service time $Pr\{S_j \geq s\}$ decreases monotonically in $j$, i.e., $S_j$ is stochastically dominated by $S_{j+1}$ for $j = 0, \dots, t-1$. To see this intuitively, type-$j$ request service start means that at exactly $j$ recovery groups, one of the forked copies departs from a leading server before the request starts service. For the completion of such requests, the single remaining copy has to finish at $j$ recovery groups with a leading server, while both of the remaining copies have to finish service at the other $t-j$ recovery groups. Waiting for one copy is stochastically better than having to wait for two, hence the larger the number of early departures the faster the request service time will be.
  Besides, since service times are non-negative we know
%   \begin{equation*}
%   \begin{split}
%     E[& S^m_j - S^m_{j+1}] \\
%     &= \int_{0}^{\infty} m s^{m-1} \bigl(Pr\{S_j \geq s\} - Pr\{S_{j+1} \geq s\}\bigr) ds > 0,
%   \end{split}
%   \end{equation*}
  \[ E[S^m_j - S^m_{j+1}] = \int_{0}^{\infty} m s^{m-1} \bigl(Pr\{S_j \geq s\} - Pr\{S_{j+1} \geq s\}\bigr) ds > 0, \]
  thus all moments of $S_j$ decrease monotonically in $j$.
%   \begin{equation}
%     \begin{split}
%       & E[S_j - S_{j+1}] = \int_{0}^{\infty} \bigl[Pr\{S_j \geq s\} - Pr\{S_{j+1} \geq s\}\bigr] ds \\
%       &= \int_{0}^{\infty} e^{-(\gamma+t\mu)s}\bigl[(2-e^{-\mu s})^{t-j} - (2-e^{-\mu s})^{t-j-1}\bigr]ds \\
%       &> 0.
%     \end{split}
%   \label{eq:eq_serv_dist_moments_mono_decrease}
%   \end{equation}
%   where the inequality follows from $2-e^{-\mu s} > 1$, and consequently
%   \[ (2-e^{-\mu s})^{t-j} - (2-e^{-\mu s})^{t-j-1} > 0 \]
%   Same calculations can be carried out to show the same holds for all the moments of $S_j$.
\end{proof}

% Service time distribution of subsequent requests are dependent.
% Service start type of a request is determined by the leveling between the servers in recovery groups. For instance in Fig.~\ref{fig:fig_reptoall_t1_servstart}, one repair server is two tasks ahead of its pair, which guarantees that request-$2$ will make type-$1$ start. Similarly, by looking at the start type of the request at HoL, one can guess leveling between the server pairs in the recovery groups. Therefore, service start type and consequently the service time distribution of the request at HoL depends on that of the departed requests.

Even though service time distributions are not independent across the requests, they are only loosely coupled. For the request at the head of the line to make type-$j$ start, there needs to be a leading server in $j$ of the recovery groups, while the servers in the rest of the recovery groups should be leveled.
Every request departure from the system triggers cancellation of the copies at the slow recovery servers, which helps the slow servers to catch up with the leading servers. Therefore, it is ``hard'' for the leading servers to keep leading because they compete not just with their slow peers but with every other server in the system. Thus, the queues at all the servers are expected to frequently level up. A time epoch at which the queues level up breaks the dependence between the service times of the requests that start service before or after the levelling. Given that these time epochs occur frequently, request service times constitute a series of independent small-size batches where the dependence exists only within each batch.
\begin{observation}
  Replicate-to-all system experiences frequent time epochs across which the request service time distributions are independent.
\label{obv_freq_renewals}
\end{observation}

% PASTA for f_j's
Lastly, request arrivals see time averages in terms of the probabilities over different service time distributions. This holds even when the service times at the servers are not exponentially distributed.
\begin{lemma}
  In replicate-to-all system, let $J_i$ be the type of service time distribution for the $i$th request. Then,
  \[ \lim_{i \to \infty} Pr\{J_i = j\} = f_j \]
  where $f_j$'s are the time average values of the probabilities over the possible request service time distributions.
\label{lm_serv_time_probs_are_time_avgs}
\end{lemma}
\begin{proof}
  Given that two requests find the system in the same state at their arrival epochs, the same stochastic process generates their life cycle under stability. This implies that the probabilities over the possible service time distributions for a request are completely determined by the system state seen by the request at its arrival. Since Poisson arrivals see time averages in system state under stability \cite{wolff1982poisson}, they also see time averages in the probabilities over the possible service time distributions.
\end{proof}

% M/G/1 approximation
Observations we have made so far, which are also validated by the simulations, allow us to develop an approximate method for analyzing the replicate-to-all system.
Requests depart in the order they arrive (Obv.~\ref{obv_reptoall}), hence the system as a whole acts like a first-in first-out queue. Thinking about the trimmed down space of possibilities or exponential service times at the servers, $t+1$ possible distributions for the request service times are given in Lm.~\ref{lm_reqservtime_dists}. Although request service times are not independent, they exhibit loose coupling (Obv.~\ref{obv_freq_renewals}). Putting all these together and relying on their accuracy gives us an $M/G/1$ approximation for the actual system.
\begin{proposition}
  Replicate-to-all system can be approximated as an $M/G/1$ queue, hence the PK formula gives an approximate value for the average hot data download time as
  \begin{equation}
    E[T] \approx E[S] + \frac{\lambda E[S^2]}{2(1 - \lambda E[S])}.
  \label{eq:eq_reptoall_PK}
  \end{equation}
  The moments of $S$ are given as
  \begin{equation}
	E[S] = \sum_{j=0}^{t} f_j E[S_j], \qquad E[S^2] = \sum_{j=0}^{t} f_j E[S_j^2],
  \label{eq:eq_reptoall_servmoments}
  \end{equation}
  where $f_j$ is the probability that an arbitrary request has type-$j$ service time distribution $S_j$ (see Lm.~\ref{lm_serv_time_probs_are_time_avgs}). When service times at the servers are exponentially distributed, moments of $S_j$ are given in \eqref{eq:eq_serv_dist_moments}.
\label{prop_reptoall_mg1}
\end{proposition}

\subsection{An exact upper and lower bound}
An upper bound on the hot data download time in replicate-to-all system is given in \cite{kadhe2015availability}. The bound is derived by assuming a more restricted service model, which is known as split-merge. One of the servers in a recovery group can execute the request copies in its queue faster than its sibling, i.e., a recovery server can lead its peer. Split-merge model does not allow this; servers in the recovery groups are blocked until the request at the head of the line departs. Thus, under the split-merge model requests are accepted into service one by one, hence all requests are restricted to have type-$0$ (slowest) service time distribution.

Employing the same idea that leads to the split-merge model, we next find a lower bound in the following.
\begin{theorem}
  An $M/G/1$ queue with service times distributed as $S_t$ in Lm.~\ref{lm_reqservtime_dists} is a lower bound on the replicate-to-all system in hot data download time.
\label{thm_reptoall_t_lb}
\end{theorem}
\begin{proof}
  In replicate-to-all system, there are $t+1$ possible request service time distributions which are stochastically ordered as $S_0 > \dots > S_t$. Restricting all request service time distributions to $S_t$ gives an $M/G/1$ queue that performs faster than the actual system.
\end{proof}

$M/G/1$ approximation given in Prop.~\ref{prop_reptoall_mg1} is an approximate intermediate point between these two polarities; split-merge upper bound and the lower bound given in Thm.~\ref{thm_reptoall_t_lb}.

% -----------------------  Conjecture  --------------------------- %
\subsection{Probabilities for the request service time distributions}
For the $M/G/1$ approximation to be useful, we need to find the probabilities $f_j$'s (equivalently their time average values) for the possible request service time distribution $S_j$'s. An exact expression for $f_j$ is found as follows. The sub-sequence of request arrivals that see an empty system forms a renewal process \cite[Thm.~5.5.8]{gallager2013stochastic}. Let us define a renewal-reward function $R_j(t) = \mathbbm{1}\{J(t) = j\}$ where $\mathbbm{1}$ is the indicator function and $\{J(t) = j\}$ denotes the event that the request at the head of the line at time $t$ has type-$j$ service time distribution $S_j$.
\begin{equation*}
  \begin{split}
    f_j =& \lim_{t \to \infty} Pr\{R_j(t) = 1\}
    = \lim_{t \to \infty} E[R_j(t)] \\
    \stackrel{\text{(a)}}{=}& \lim_{t \to \infty} \frac{1}{t} \int_{-\infty}^t R_j(\tau) d\tau \stackrel{\text{(b)}}{=} \frac{E\Bigl[\int_{S_{n-1}^r}^{S_n^r}R_j(t)dt\Bigr]}{E[X]}.
  \end{split}
\end{equation*}
where $(a)$ and $(b)$ are due to the equality of the limiting time and ensemble averages of the renewal-reward function $R_j(t)$ \cite[Thm.~5.4.5]{gallager2013stochastic}, and $S_{n-1}^r$, $S_n^r$ are the $(n-1)$th, $n$th renewal epochs (i.e., consecutive arrival epochs that find the system empty), and $X$ is the i.i.d. inter-renewal interval.

As the expression above clearly indicates, deriving $f_j$'s requires an exact analysis of the system. However, we conjecture the following relation that allows us to find good estimates rather than the exact values of $f_j$'s which are presented in the following sections.
\begin{conjecture}
  In replicate-to-all system, probability $f_j$'s over the trimmed down space of the possible request service time distribution $S_j$'s (see Lm.~\ref{lm_reqservtime_dists}) hold the relation
  \[ f_{i-1} > f_i, \quad 1 \leq i \leq t. \]
  or equivalently, $f_i = \rho_i f_{i-1}$ for some $\rho_i < 1$.
\label{conj_reptoall_t_fjs}
\end{conjecture}
We here briefly discuss the reasoning behind the conjecture. Obv.~\ref{obv_freq_renewals} states that the servers frequently level up because the leading servers in the recovery groups compete with every other server to keep leading or to possibly advance further ahead. For a request to be served according to type-$j$ distribution, one of the forked copies of the request has to depart at exactly $j$ recovery groups before the request starts service. This requires one server in each of the $j$ recovery groups to be leading, which is harder for larger values of $j$.
We validated the conjecture with the simulations but could not prove it. However, we derive a strong pointer for the conjecture in Appendix~\ref{subsec:subsec_simplex_t_conjecture}; given a request is served according to type-$j$ distribution, the next request in line is more likely to be served according to type-$(<j)$ distribution.

% ######################################  Simplex(t=1)  ######################################## %
\section{Replicate-to-all with Availability One}
\label{sec:sec_reptoall_t1}
In this section, we study replicate-to-all system with availability one, e.g., the storage system $[a, b, a+b]$ as illustrated in Fig.~\ref{fig:fig_reptoall_t1_servstart}. We initially assume that service times at the servers are exponentially distributed and let the service rates at the systematic server and the two recovery servers be respectively $\gamma$, $\alpha$ and $\beta$. Then, there are two possible request service time distributions $S_0$ and $S_1$ as given in Lm.~\ref{lm_reqservtime_dists}. $M/G/1$  approximation in Prop.~\ref{prop_reptoall_mg1} requires finding the probabilities $f_0$ and $f_1$ for an arbitrary request to be served according to $S_0$ and $S_1$ respectively. Although an exact solution proves to be hard, we find good estimates for $f_0$ and $f_1$ in the following.

% ---------------------------------  Markov process  ----------------------------- %
\subsection{Markov process for the system state}
Imagine a join queue attached at the tail of the system that enqueues the request copies that depart from the servers. Since a copy departing from the systematic server completes the associated request immediately, a copy waiting for synchronization in the join queue must be from a leading recovery server. As soon as a request completes, copies of the request waiting in the join queue are removed. State of the join queue at time $t$ is the tuple $\bm{n(t)} = (n_1(t), n_2(t))$ where $n_i$'s denote the number of request copies in the join queue that departed from the two recovery servers respectively.

\begin{figure}[t]
  \begin{center}
    \includegraphics[width=0.5\textwidth, keepaspectratio=true]{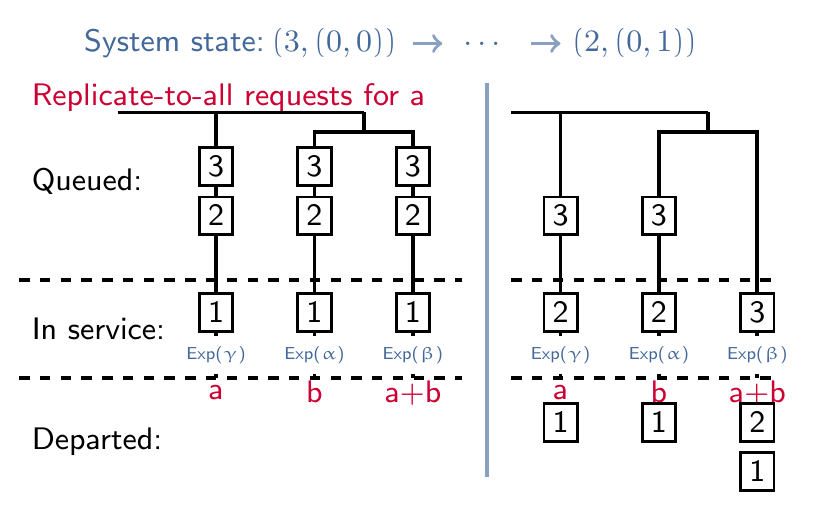}
  \end{center}
  \vspace*{-0.35cm}
  \caption{Replicate-to-all system with availability one. Requests $1$-$3$ are replicated to the systematic server and the recovery servers at arrival. Two system snapshots are illustrated for time epochs at which the request-$1$ (Left) and the request-$2$ (Right) starts service. Request-$1$ has service time distributed as $S_0$ while request-$2$ is served with $S_1$. Request service time distributions $S_0$ and $S_1$ are defined in Lm.~\ref{lm_reqservtime_dists}.}
  \label{fig:fig_reptoall_t1_servstart}
\end{figure}

Ordered departure of the requests (see Obv.~\ref{obv_reptoall}) and cancellation of the outstanding copies upon a request completion imply $n_1(t)n_2(t) = 0$ for all $t$. Together with the total number of requests $N(t)$ in the system; $\bm{s(t)} = (N(t), n_1(t), n_2(t))$ constitutes the state of the replicate-to-all system at time $t$ (see Fig.~\ref{fig:fig_reptoall_t1_servstart}). System state $\bm{s(t)}$ is a Markov process as illustrated in Fig.~\ref{fig:fig_reptoall_t1_mp__high_traff_approx}. Let us define $Pr\{\bm{s(t)} = (k, i, j)\} = p_{k,i,j}(t)$, and suppose that stability is imposed and $\lim_{t \to \infty} p_{k,i,j}(t) = p_{k, i, j}$. Then, the system balance equations are summarized for $k,i,j \geq 0$ as
% \begin{equation}
% \begin{split}
%   \bigl[\gamma & \mathbbm{1}(k \geq 1) + \alpha \mathbbm{1}(i \geq 1) + \beta \mathbbm{1}(j \geq 1)\bigr]p_{k,i,j} = \\
%           & \lambda \mathbbm{1}(k \geq 1, i \geq 1, j \geq 1)p_{k-1,i-1,j-1} + \\
%           & \gamma p_{k+1,i+1,j+1} + (\gamma + \alpha)p_{k+1,i+1,j} + (\gamma + \beta)p_{k+1,i,j+1}.
% \end{split}
% \label{eq:eq_reptoall_t1_balance_eqs}
% \end{equation}
\begin{equation}
\begin{split}
  \bigl[\gamma \mathbbm{1}(k \geq 1) + \alpha \mathbbm{1}(i \geq 1) + \beta \mathbbm{1}(j \geq 1)\bigr]p_{k,i,j} &=
          \lambda \mathbbm{1}(k \geq 1, i \geq 1, j \geq 1)p_{k-1,i-1,j-1} \\
          &\quad + \gamma p_{k+1,i+1,j+1} + (\gamma + \alpha)p_{k+1,i+1,j} + (\gamma + \beta)p_{k+1,i,j+1}.
\end{split}
\label{eq:eq_reptoall_t1_balance_eqs}
\end{equation}
Computing the generating function
\[ P_{w,x,y} = \sum_{k, i, j > 0} p_{k,i,j}w^k x^i y^j \]
from the above balance equations is intractable, so the exact solution of the system's steady state behavior is. This is because the state space is infinite in two dimensions and its analysis is tedious (see Appendix~\ref{subsec:subsec_reptoall_t1_pyramid_mp_analysis} for a guessing based analysis with local balance method). As discussed next, state space of the system can be made much simpler by employing a high traffic assumption.
\begin{figure}[htb]
  \centering
  \begin{tikzpicture}
    % \node at (0,2.5) {\includegraphics[scale=0.75]{MCpyramid}};
    % \node at (0,0) {\includegraphics[scale=0.75]{S21JoinHighT}};
    \node at (0,2.5) {\includegraphics[scale=0.85]{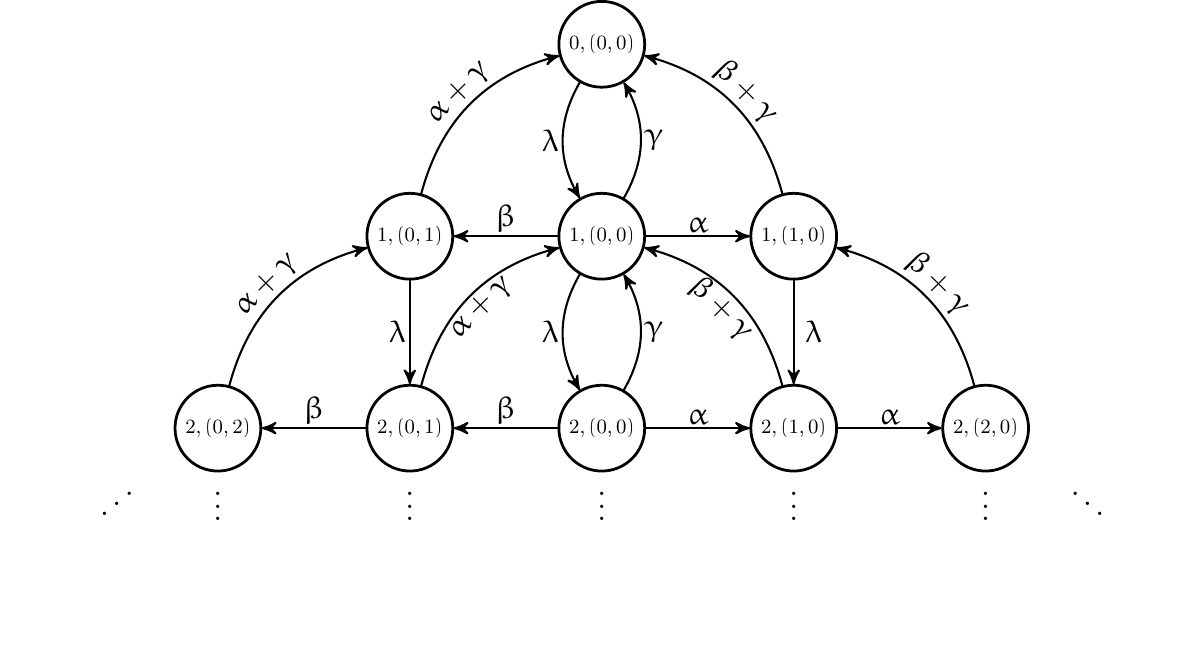}};
    \node at (0,0) {\includegraphics[scale=0.85]{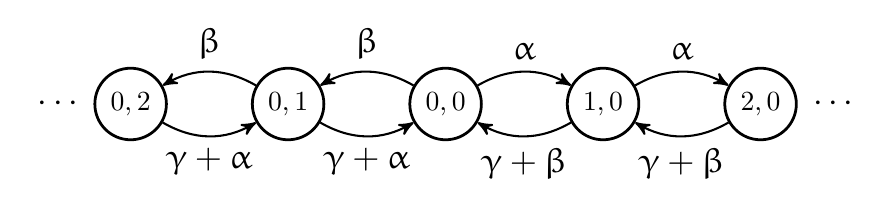}};
  \end{tikzpicture}
  \caption{Markov state process for the replicate-to-all system with availability one (Top), and its high traffic approximation (Bottom).}
  \label{fig:fig_reptoall_t1_mp__high_traff_approx}
\end{figure}

% ---------------------------------  High-traffic assumption  ----------------------------- %
\subsection{Analysis under high traffic assumption}
\label{subsec:subsec_reptoall_t1_hightraff}
Suppose that the arrival rate $\lambda$ is very close to its critical value for stability and the queues at the servers are always nonempty. This is a rather crude assumption and holds only when the system is unstable, however it yields a good approximation for the system and makes the analysis much easier. We refer to this set of working conditions as \textit{high traffic assumption}. Under this assumption, servers are always busy and the state of the join queue $\bm{n(t)}$ represents the state of the whole system. Then, the system state follows a birth-death process as shown in Fig.~\ref{fig:fig_reptoall_t1_mp__high_traff_approx}, for which the steady state balance equations are given as
\begin{equation}
  \alpha p_{i,0} = (\gamma + \beta) p_{i+1,0}, \quad\quad
  \beta p_{0,i} = (\gamma + \alpha) p_{0,i+1}, ~i \geq 0.
\label{eq:eq_simplest_simplex_join_balance_high_traffic}
\end{equation}
where $p_{i,j} = \lim_{t \to \infty} Pr\{\bm{n(t)} = (i, j)\}$. Solving the balance equations, we find
\begin{equation}
  p_{i,0} = \left(\frac{\alpha}{\beta + \gamma}\right)^i p_{0,0}, \qquad
  p_{0,i} = \left(\frac{\beta}{\alpha + \gamma}\right)^i p_{0,0}, ~i \geq 1.
\label{eq:eq_simplest_simplex_steady_state}
\end{equation}
By the axiom of probability, we have
\[ \left[1 + \sum_{i=1}^{\infty} \left(\frac{\beta}{\alpha + \gamma}\right)^i + \left(\frac{\alpha}{\beta + \gamma}\right)^i \right] p_{0,0} = 1. \]
Assuming $\beta < \alpha + \gamma$ and $\alpha < \beta + \gamma$,
\[ p_{0,0} = \left(1 +  \frac{\beta}{\alpha + \gamma - \beta} + \frac{\alpha}{\beta + \gamma - \alpha}\right)^{-1} = \frac{\gamma^2-(\alpha-\beta)^2}{\gamma(\alpha+\beta+\gamma)} \]
from which $p_{i,0}$ and $p_{0,i}$ are found by substituting $p_{0,0}$ in \eqref{eq:eq_simplest_simplex_steady_state}.
Other useful quantities are given as
\[ \sum_{i=1}^\infty p_{i,0} = \frac{\alpha(\alpha+\gamma-\beta)}{\gamma(\alpha+\beta+\gamma)}, \quad\quad \sum_{i=1}^\infty p_{0,i} = \frac{\beta(\beta+\gamma-\alpha)}{\gamma(\alpha+\beta+\gamma)}. \]
For tractability, we continue by assuming $\alpha = \beta = \mu$, thus
\[ p_{0,0} = \frac{\gamma}{\gamma+2\mu}, \quad \sum_{i=1}^{\infty} p_{i,0} = \sum_{i=1}^{\infty} p_{0,i} = \frac{\mu}{\gamma+2\mu}. \]
Next, we use the steady state probabilities under the high traffic assumption to obtain estimates for the request service time probabilities $f_0$ and $f_1$, and some other quantities that give insight into the system behavior.

% ---------------------------------  Winning frequencies  ----------------------------- %
\subsection{Request completions at the servers}
\label{subsec:subsec_reptoall_t1_winning_freqs}
A request completes as soon as either its copy at the systematic server departs or both copies that it has at the recovery group. Here we find bounds on the fraction of the requests completed by the systematic server or the recovery group.
\begin{theorem}
  In replicate-to-all system with availability one, let $w_s$ and $w_r$ be the fraction of requests completed by respectively the systematic server and the recovery group. Then the following inequalities hold
  \begin{equation}
    w_s \geq \frac{\gamma\nu}{\gamma\nu+2\mu^2}, \quad
    w_r \leq \frac{\mu^2}{\gamma\nu+2\mu^2}; \quad \nu = \gamma+2\mu.
  \label{eq:eq_reptoall_t1_winning_freqs}
  \end{equation}
  \label{thm_reptoall_t1_winning_freqs}
\end{theorem}
\begin{proof}
  Under high traffic approximation, $w_s$ and $w_r$ can be found from the steady state probabilities of the Markov chain (MC) embedded in the state process (see Fig.~\ref{fig:fig_reptoall_t1_mp__high_traff_approx}). Recall that the system state under high traffic approximation is $\bm{n(t)} = (n_1(t), n_2(t))$ where $n_i$ is the number of request copies in the join queue that departed from the $i$th recovery server.
  
  System stays at each state for an exponential duration of rate $\nu = \gamma + 2\mu$. Therefore, the steady state probabilities $p_i$'s (i.e., limiting fraction of the time spent in state $i$) of $\bm{n(t)}$ and the steady state probabilities $\pi_i$'s (i.e., limiting fraction of the state transitions into state $i$) of the embedded MC are equal as seen by the equality
  \[ \pi_i = (p_i\nu) / \sum_{i \geq 0} p_i\nu = p_i. \]
  
  Let $f_s$ and $f_r$ be the limiting fraction of state transitions that represent request completions by respectively the systematic server and the recovery group. Then, we find
  \begin{align*}
    & f_s = \pi_{0,0}\frac{\gamma}{\nu} + \sum_{i=1}^{\infty} \pi_{i,0}\frac{\gamma}{\nu} + \sum_{i=1}^{\infty} \pi_{0,i}\frac{\gamma}{\nu} = \frac{\gamma}{\nu}, \\
    & f_r = \sum_{i=1}^{\infty} (\pi_{0,i} + \pi_{i,0})\frac{\mu}{\nu} = 2\left(\frac{\mu}{\nu}\right)^2.
  \end{align*}
  Limiting fraction of all state transitions that represent a request departure is then found as
  \[ f_d = f_s + f_r = (\gamma\nu+2\mu^2)/\nu^2. \]
  
  Finally, we find the fraction of request completions by the systematic server and by the recovery group as
  \[
    \hat{w}_s = \frac{f_s}{f_d} = \frac{\gamma\nu}{\gamma\nu+2\mu^2}, \quad
    \hat{w}_r = \frac{f_r}{f_d} = \frac{2\mu^2}{\gamma\nu+2\mu^2}.
  \]
  Recall that, these values are calculated using the high traffic approximation, in which the queues are never empty. However, recovery servers have to regularly go idle under stability. Therefore, the fraction of request completions at the recovery group is smaller under stability than what we find under high traffic, hence we conclude $w_r \leq \hat{w}_r$, which implies $w_s \geq \hat{w}_s$.
\end{proof}

Bounds on $w_s$ and $w_r$ become tighter as the arrival rate $\lambda$ increases as shown in Fig.~\ref{fig:plot_winning_freqs}, which is because the bounds are derived by studying the system under the high traffic assumption.
\begin{figure}[hbt]
  \centering
  \includegraphics[width=0.45\textwidth, keepaspectratio=true]{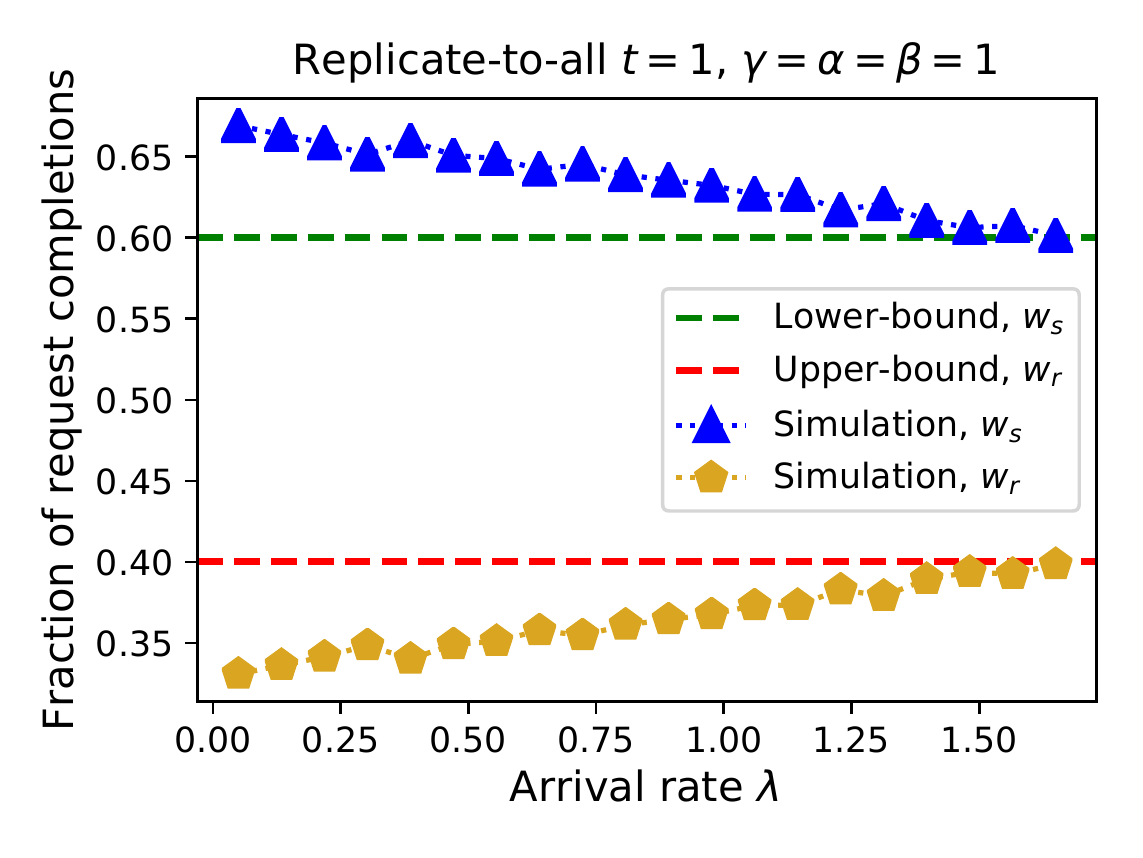}
  \caption{Simulated fraction of request completions by the systematic server ($w_s$) and by the recovery group ($w_r$). Bounds are as given in \eqref{eq:eq_reptoall_t1_winning_freqs}.}
\label{fig:plot_winning_freqs}
\end{figure}

% ---------------------------------  Average system time  ------------------------ %
\subsection{$M/G/1$ approximation}
\label{subsec:subsec_reptoall_t1_mg1approx}
Using the analysis under high traffic assumption given in Sec.~\ref{subsec:subsec_reptoall_t1_hightraff}, here we obtain estimates for the probabilities $f_0$ and $f_1$ over the request service time distributions $S_0$ and $S_1$. This enables us to state an analytical form of the $M/G/1$ approximation.

\begin{lemma}
  Hot data download time under high traffic assumption is a lower bound for the download time in the replicate-to-all system.
\label{lm_reptoall_t1_lb}
\end{lemma}
\begin{proof}
  Let us first consider the system with availability one. Comparing the two Markov processes in Fig.~\ref{fig:fig_reptoall_t1_mp__high_traff_approx}, one can see that the process under high traffic assumption can be obtained from the actual state process as follows: 1) Introduce additional transitions of rate $\alpha$ from state $(i,(i,0))$ to $(i+1,(i+1,0))$ and additional transitions of rate $\beta$ from state $(i,(0,i))$ to $(i+1,(0,i+1))$ for each $i \geq 0$, 2) Gather the states $(i,(m,n))$ for $i \geq 0$ into a super state, (3) Observe that the process with the super states is the same as the process under high traffic assumption. Thus, employing the high traffic assumption has the effect of placing extra state transitions that help the system to serve the requests faster.
  
  The above rationale can be extended for the system with any degree of availability. Recall from Lm.~\ref{lm_reqservtime_dists} that request service time distributions are stochastically ordered as $S_0 > \dots > S_t$. High traffic assumption ensures that there is always a replica for the leading recovery servers to serve, hence increases the fraction of requests served with a faster request service time distribution. Thus, the analysis with high traffic assumption would yield a lower bound on the download time in the actual system.
\end{proof}

\begin{theorem}
  Under the assumptions in Prop.~\ref{prop_reptoall_mg1} and for replicate-to-all system with availability one, we have the following bounds on the probabilities $f_0$ and $f_1$ over the request service time distributions $S_0$ and $S_1$ (see Lm.~\ref{lm_reqservtime_dists})
  \[ f_0 \geq f^{lb}_0 = \frac{\gamma\nu}{\gamma\nu+2\mu^2}, \qquad f_1 \leq f^{ub}_1 = \frac{2\mu^2}{\gamma\nu+2\mu^2} \]
\label{thm_reptoall_t1_hightraff_approx}
\end{theorem}
\begin{proof}
%   Markov process for system state under high-traffic assumption is given in Fig.~\ref{fig:fig_reptoall_t1_mp__high_traff_approx}.
  Consider the state process under high traffic assumption as given in Fig.~\ref{fig:fig_reptoall_t1_mp__high_traff_approx}. State transitions that are towards the center $(0,0)$ state represent request completions. Let us define $f_d$ as the fraction of such state transitions.
  
  Since queues are never empty under high traffic assumption, the next request always starts service as soon as the head of the line request departs. Therefore, among the state transitions that represent request completions, a request makes a type-$0$ service start per transition into state $(0,0)$ while a request makes a type-$1$ service start per transition to any other states. Let $f_{\to 0}$ and $f_{\to 1}$ be the fraction of state transitions that represent respectively type-$0$ and type-$1$ request service starts. We find
  \begin{align*}
    & f_d = \pi_{0,0}\frac{\gamma}{\nu} + \sum_{i=1}^{\infty} (\pi_{0,i} + \pi_{i,0})\frac{\mu+\gamma}{\nu} = \frac{2\mu^2+2\mu\gamma+\gamma^2}{\nu^2} \\
    & \begin{aligned}
        f_{\to 0} &= \pi_{0,0}\frac{\gamma}{\nu} + \pi_{1,0}\frac{\mu + \gamma}{\nu} + \pi_{0,1}\frac{\mu + \gamma}{\nu} \\
        &= \pi_{0,0}\left(\frac{\gamma}{\nu} + \frac{2\mu}{\mu + \gamma}\frac{\mu + \gamma}{\nu}\right) = \pi_{0,0} = \frac{\gamma}{\gamma+2\mu},
      \end{aligned}
  \end{align*}
  
  Then, under high traffic approximation, the limiting fraction of the requests that make type-$0$ ($\hat{f}_0$) or type-$1$ ($\hat{f}_1$) service start are found for $\nu = \gamma + 2\mu$ as
  \[ \hat{f}_0 = \frac{f_{\to 0}}{f_d} = \frac{\gamma\nu}{\gamma\nu+2\mu^2}, \qquad \hat{f}_1 = 1 - \hat{f}_0 = \frac{2\mu^2}{\gamma\nu+2\mu^2}. \]
  
  Under stability, system has to empty out regularly. A request that finds the system empty makes type-$0$ service start for sure. However under high traffic assumption, servers are always busy, which enforces less number of type-$0$ request service starts. Therefore, the fraction $\hat{f}_0$ of type-$0$ request service starts obtained under high traffic approximation presents a lower bound for its value under stability, i.e., $\hat{f}_0 \leq f_0$, from which $\hat{f}_1 \geq f_1$ follows immediately.
\end{proof}

Given that type-$0$ service is slower than type-$1$ service, substituting the bounds $f^{lb}_0$ and $f^{ub}_1$ given in Thm.~\ref{thm_reptoall_t1_hightraff_approx} in place of the actual probabilities $f_0$ and $f_1$ in \eqref{eq:eq_reptoall_servmoments} yields the lower bounds on the request service time moments as
\begin{equation}
\begin{split}
  E[S_{lb}] &= f^{lb}_0\left(\frac{2}{\gamma+\mu} - \frac{1}{\gamma+2\mu}\right) + f^{ub}_1\frac{1}{\gamma+\mu}, \\
  E[S_{lb}^2] &= f^{lb}_0\left(\frac{4}{(\gamma+\mu)^2} - \frac{2}{(\gamma+2\mu)^2}\right) + f^{ub}_1\frac{2}{(\gamma+\mu)^2}.
\end{split}
\label{eq:eq_simplex_t_1_serv_time_moments_values}
\end{equation}
As suggested by the $M/G/1$ approximation in Prop.~\ref{prop_reptoall_mg1}, substituting these lower bounds on the service time moments in the PK formula gives an \textit{approximate} lower bound on the hot data download time in replicate-to-all system with availability one. To emphasize, since the $M/G/1$ model is not exact but an approximation, the lower bound on the download time can only be claimed to be an approximate one. Simulated average hot data download times and the derived $M/G/1$ approximation are compared in Fig.~\ref{fig:fig_reptoall_t1_sim_vs_model}. Approximation seems to be doing very good at predicting the actual average download time, and especially better than the split-merge upper bound given in \cite{kadhe2015availability} and the lower bound we stated in Thm.~\ref{thm_reptoall_t_lb}.
\begin{figure}[t]
  \centering
  \includegraphics[width=0.45\textwidth, keepaspectratio=true]{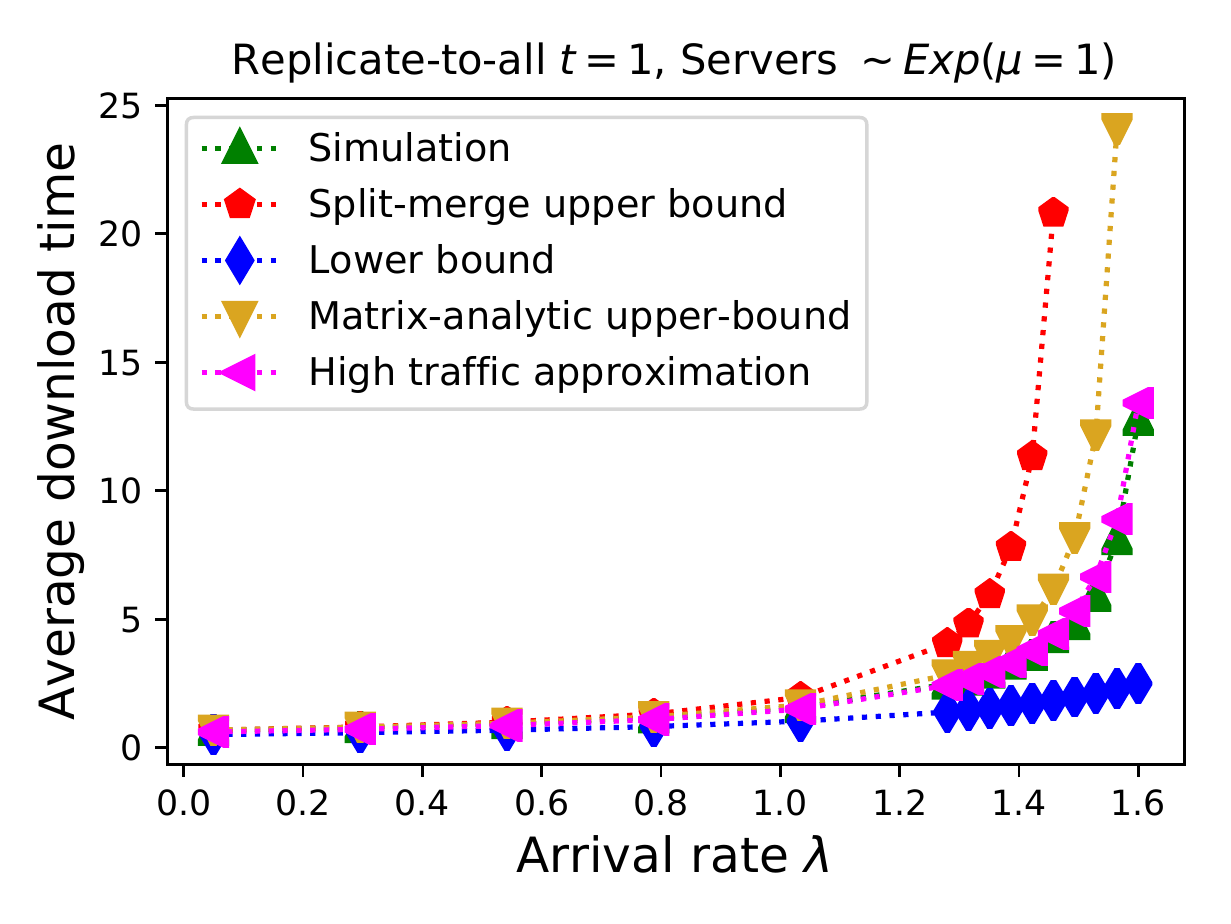}
  \caption{For replicate-to-all system with availability one and exponentially distributed service times at the servers, comparison of the simulated average hot data download time, the split-merge upper bound in \cite{kadhe2015availability}, the lower bound in Thm.\ref{thm_reptoall_t_lb}, matrix analytic upper bound in Thm.~\ref{thm_reptoall_t1_matrixanalytic_ub}, and the $M/G/1$ approximation given in Prop.~\ref{prop_reptoall_mg1} and Thm.~\ref{thm_reptoall_t1_hightraff_approx}.}
\label{fig:fig_reptoall_t1_sim_vs_model}
\end{figure}

% ---------------------------------  Resource allocation  ------------------------ %
\subsection{Service rate allocation}
\label{subsec:subsec_reptoall_t1_resalloc}
Suppose that the system is given a budget of service rate which can be arbitrarily allocated across the servers. When service times at the server are exponentially distributed, we show that allocating more service rate at the systematic server achieves smaller hot data download time (see Appendix~\ref{subsec:subsec_dE_T_dro_less_than_zero}). This suggests that in systems employing availability codes for reliable storage, systematic symbols shall be stored on the fast nodes while the slow nodes can be used for storing the redundant symbols for faster hot data download.

% -------------------------  Matrix analytic solution for Simplex(t=1)  ------------------------------- %
\subsection{Matrix analytic solution}
Here we find an upper bound on the average hot data download time that is provably tighter than the split-merge upper bound in \cite{kadhe2015availability}. Let us truncate the state process for the replicate-to-all system with availability one (see Fig.~\ref{fig:fig_reptoall_t1_mp__high_traff_approx}) such that the pyramid is limited to only the five central columns and infinite only in the vertical dimension. This choice is intuitive given Conj.~\ref{conj_reptoall_t_fjs}; system spends more time at the central states which implement type-$0$ request and it is hard for the leading servers to advance further ahead, which would have made the system state transit towards the wings of the pyramid Markov Chain.
Also the analysis in Appendix~\ref{subsec:subsec_reptoall_t1_pyramid_mp_analysis} suggests that the most frequently visited system states are located at the central columns.

\subsubsection{Computing the steady state probabilities}
Finding a closed form solution for the steady state probabilities of the states in the truncated process is as challenging as the original problem. However one can solve the truncated process numerically with an arbitrarily small error using the \textit{Matrix Analytics method} \cite[Chapter~21]{MorHB13}.
In the following, we denote the vectors and matrices in bold font. Let us first define the limiting steady state probability vector as
\begin{equation*}
  \begin{split}
    \bm{\pi} =& [\pi_{0,(0,0)}, \pi_{1,(0,1)}, \pi_{1,(0,0)}, \pi_{1,(1,0)}, \\
      & \quad \pi_{2,(0,2)}, \pi_{2,(0,1)}, \pi_{2,(0,0)}, \pi_{2,(0,1)}, \pi_{2,(2,0)}, \\
	    & \quad  \pi_{3,(0,2)}, \pi_{3,(0,1)}, \pi_{3,(0,0)}, \pi_{3,(1,0)} ,\pi_{3,(2,0)}, \cdots] \\
     =& [\bm{\pi}_0, \bm{\pi}_1, \bm{\pi}_2, \bm{\pi}_3, \bm{\pi}_4, \cdots].
  \end{split}
\end{equation*}
where $\pi_{k,(i,j)}$ is the steady state probability for state $(k,(i,j))$ and
\begin{equation*}
  \begin{split}
    & \bm{\pi}_0 = [\pi_{0,(0,0)}, \pi_{1,(0,1)}, \pi_{1,(0,0)}, \pi_{1,(1,0)}], \\
    & \bm{\pi}_i = [\pi_{i+1,(0,2)}, \pi_{i+1,(0,1)}, \pi_{i+1,(0,0)}, \pi_{i+1,(1,0)}, \pi_{i+1,(2,0)}].
  \end{split}
\end{equation*}
for $i \geq 1$. One can write the balance equations governing the limiting probabilities in the form below
\begin{equation}
  \label{eq:eq_balance_eq}
  \bm{\pi}\bm{Q} = \bm{0}
\end{equation}
For the truncated process,  $\bm{Q}$ has the following form
\begin{equation}
  \bm{Q}   =
  \left[\begin{array}{ccccccc}
  \bm{F}_0 & \bm{H}_0    &        &        &        &             & \\
  \bm{L}_0 & \bm{F}	     & \bm{H} &        &        & \bm{0}      & \\
           & \bm{L}	     & \bm{F} & \bm{H} &        &             & \\
           &             & \bm{L} & \bm{F} & \bm{H} &             & \\
           & \bm{0}      &  	    & \bm{L} & \bm{F} & \bm{H}      & \\
           &	           &	      &	     & \ddots   & \ddots      & \\
  \end{array} \right],
\label{eq:eq_Q}
\end{equation}
where the sub-matrices $\bm{F}_0$, $\bm{H}_0$, $\bm{L}_0$, $\bm{F}$, $\bm{L}$ and $\bm{H}$ are given in Appendix~\ref{subsec:subsec_reptoall_t1_matrix_analytic} in terms of server service rates $\gamma$, $\alpha$, $\beta$ and the request arrival rate $\lambda$. Using \eqref{eq:eq_balance_eq} and \eqref{eq:eq_Q}, we get the following system of equations in matrix form,
\begin{equation}
  \begin{split}
    & \bm{\pi}_0\bm{F}_0 + \bm{\pi}_1\bm{L}_0 = 0, \\
  	& \bm{\pi}_0\bm{H}_0 + \bm{\pi}_1\bm{F} + \bm{\pi}_2\bm{L} = 0, \\
  	& \bm{\pi}_i\bm{H} + \bm{\pi}_{i+1}\bm{F} + \bm{\pi}_{i+2}\bm{L} = 0, \quad i \geq 1.
  \end{split}
\label{eq:eq_system}
\end{equation}
% \begin{figure}
%   \begin{center}
%   	\includegraphics[scale=1,trim= 1.5cm 1cm 1.5cm 0cm]{MatrixAnalytics}
%     \caption{Truncated Markov chain corresponding to the problem in Fig.~\ref{fig:fig_reptoall_t1_mp__high_traff_approx}}
%     \label{fig:fig_simplex_t_1_truncated_mp_matrix_analytics}
%   \end{center}
% \end{figure}
In order to solve the system above, we assume the steady state probability vectors to be of the form
\begin{equation}
  \bm{\pi}_i = \bm{\pi}_1\bm{R}^{i-1}, \quad i \geq 1, \bm{R} \in R^{5 \times 5}.
\label{eq:eq_pi}
\end{equation}
Combining \eqref{eq:eq_system} and \eqref{eq:eq_pi}, we get
\begin{equation}
  \begin{split}
  & \bm{\pi}_0\bm{F}_0 + \bm{\pi}_1\bm{L}_0 = 0, \\
  & \bm{\pi}_0\bm{H}_0 + \bm{\pi}_1(\bm{F} + \bm{RL}) = 0, \\
  & \bm{\pi}_i(\bm{H} + \bm{RF} + \bm{R}^2\bm{L}) = 0, \quad i \geq 1.
  \end{split}
\label{eq:eq_system2}
\end{equation}
From \eqref{eq:eq_system2} we have common conditions for the system to hold
\begin{equation}
\bm{H}+\bm{RF}+\bm{R}^2\bm{L} = \bm{0},\qquad \bm{R} = -\left(\bm{R}^2\bm{L}+\bm{H}\right)\bm{F}^{-1}.
\label{eq:eq_R}
\end{equation}

The inverse of $\bm{F}$ in \eqref{eq:eq_R} exists since $\det(\bm{F}) = -\delta^3(\delta-\alpha)(\delta-\beta) \neq 0$ assuming $\delta=\alpha+\beta+\gamma+\lambda$ and $\lambda > 0$. Using \eqref{eq:eq_R}, an iterative algorithm to compute $\bm{R}$ is given in Algorithm~\ref{alg:alg_R}. The norm $\norm{\bm{R}_i-\bm{R}_{i-1}}$ corresponds to the absolute value of the largest element of the difference matrix $\bm{R}_i-\bm{R}_{i-1}$. Therefore, the algorithm terminates when the largest difference between the elements of the last two computed matrices is smaller than the threshold $\epsilon$. The initial matrix $\bm{R}_0$ could take any value, not necessarily $\bm{0}$. The error threshold $\epsilon$ could be fixed to any arbitrary value, but the lower this value the slower the convergence.
\begin{algorithm}
\caption{Computing matrix $\bm{R}$}
\label{alg:alg_R}
  \begin{algorithmic}[1]
  \Procedure{ComputingR}{}
    \State $\epsilon \gets 10^{-6}$, $\bm{R}_0 \gets \bm{0}$, ${i} \gets 1$
    \While{\textbf{true}}
      \State $\bm{R}_i \gets -\left(\bm{R}_{i-1}^2\bm{L}+\bm{H}\right)\bm{F}^{-1}$
      \If{$\norm{\bm{R}_i-\bm{R}_{i-1}} > \epsilon$}
        \State ${i} \gets {i+1}$
      \Else\ \Return $\bm{R}_i$
      \EndIf
    \EndWhile
  \EndProcedure
  \end{algorithmic}
\end{algorithm}
Computing $\bm{R}$, vectors $\bm{\pi}_0$ and $\bm{\pi}_1$ are remaining to be found in order to deduce the values of all limiting probabilities. Recall that in \eqref{eq:eq_system2}, the first two equations are yet to be used. Writing these two equations in matrix form
\begin{equation}
  \left[\begin{array}{cc}
    \bm{\pi}_0 & \bm{\pi}_1
    \end{array}\right]\left[\begin{array}{cc}
    \bm{F}_0 & \bm{H}_0\\
    \bm{L}_0 & \bm{RL}+\bm{F}
  \end{array}\right]=\bm{0},
\label{eq:eq_finding_pi0_pi1_1}
\end{equation}
where $\bm{0}$ is a $1 \times 9$ zeros vector and
\[ \Phi=\left[\begin{array}{cc}
\bm{F}_0 & \bm{H}_0 \\
\bm{L}_0 & \bm{RL}+\bm{F}
\end{array}\right] \in R^{9 \times 9}. \]
In addition, we have the normalization equation to take into account. Denoting $\bm{1}_0=[1, 1, 1, 1]$, $\bm{1}_1=[1, 1, 1, 1, 1]$ and using \eqref{eq:eq_pi}, we get
\begin{align}
  \bm{\pi}_0\bm{1}_0^D + \sum_{i=1}^\infty \bm{\pi}_i\bm{1}_1^D &= 1 \nonumber\\
  \bm{\pi}_0\bm{1}_0^D + \sum_{i=1}^\infty \bm{\pi}_1\bm{R}^{i-1}\bm{1}_1^D &= 1\nonumber \\
  %\bm{\pi}_0\bm{1}_0^D + \bm{\pi}_1(\sum_{i=1}^\infty \bm{R}^{i-1})\bm{1}_1^D &= 1 \\
  \bm{\pi}_0\bm{1}_0^D + \bm{\pi}_1(\bm{I}-\bm{R})^{-1}\bm{1}_1^D &= 1\nonumber \\
  \left[\begin{array}{cc}
    \bm{\pi}_0 & \bm{\pi}_1
  \end{array}\right]\left[\begin{array}{c}
    \bm{1}_0^D \\
    \left(\bm{I}-\bm{R}\right)^{-1}.\bm{1}_1^D
  \end{array}\right] &= 1,
\label{eq:eq_norm}
\end{align}
where $\bm{I}$ is the $5 \times 5$ identity matrix. In order to find $\bm{\pi}_0$ and $\bm{\pi}_1$, we solve the following system
\begin{align}
  \left[\begin{array}{cc}
    \bm{\pi}_0 & \bm{\pi}_1
  \end{array}\right]\Psi &= \left[1, 0, 0, 0, 0, 0, 0, 0, 0\right].
\label{eq:eq_system_final}
\end{align}
where $\Psi$ is obtained by replacing the first column of $\Phi$ with $[\bm{1}_0, \bm{1}_1(\bm{I}-\bm{R}^D)^{-1}]^D$. Hence, \eqref{eq:eq_system_final} is a linear system of $9$ equations with $9$ unknowns. After solving \eqref{eq:eq_system_final}, we obtain the remaining limiting probabilities vector using \eqref{eq:eq_pi}.

\subsubsection{Bounding the average hot data download time}
Let $N_{ma}$ be the number of requests in the truncated system. First notice that
\begin{equation*}
  \begin{split}
    Pr\{N_{ma} = 0\} &= \pi_{0,(0,0)}, \\
	  Pr\{N_{ma} = 1\} &= \pi_{1,(0,0)} + \pi_{1,(0,1)} + \pi_{1,(1,0)} \\
	  &= \bm{\pi}_0\bm{1}_0^D-\pi_{0,(0,0)}, \\
	  Pr\{N_{ma} = i\} &= \pi_{i,(0,2)} + \pi_{i,(0,1)} + \pi_{i,(0,0)} \\
	  &\quad + \pi_{i,(1,0)} + \pi_{i,(2,0)} = \bm{\pi}_{i-1}\bm{1}_1^D, \quad i \geq 2.
  \end{split}
\end{equation*}
Then, the average number of requests in the truncated system is computed as
\begin{equation}
  \begin{split}
E[&N_{ma}] = \textstyle \sum_{i=0}^{\infty} i Pr\{N_{ma}=i\} \\
    &= \bm{\pi}_0\bm{1}_0^D-\pi_{0,(0,0)} + \sum_{i=2}^{\infty} i(\bm{\pi}_{i-1}\bm{1}_1^D) \\
    &= \bm{\pi}_0\bm{1}_0^D-\pi_{0,(0,0)} + \sum_{i=2}^{\infty} i(\bm{\pi}_1\bm{R}^{i-2}\bm{1}_1^D) \\
    % &= \bm{\pi}_0\bm{1}_0^D-\pi_{0,(0,0)} + \bm{\pi}_1 (\sum_{i=2}^{\infty} i \bm{R}^{i-2})\bm{1}_1^D \\
    &= \bm{\pi}_0\bm{1}_0^D-\pi_{0,(0,0)} + \bm{\pi}_1 \left(\sum_{i=2}^{\infty} (i-1)\bm{R}^{i-2} + \bm{R}^{i-2}\right)\bm{1}_1^D \\
    &= \bm{\pi}_0\bm{1}_0^D-\pi_{0,(0,0)} + \bm{\pi}_1 \left(\sum_{j=1}^{\infty} j\bm{R}^{j-1} + \sum_{i=0}^{\infty} \bm{R}^i\right)\bm{1}_1^D \\
    &= \bm{\pi}_0\bm{1}_0^D-\pi_{0,(0,0)} + \bm{\pi}_1 \left((\bm{I}-\bm{R})^{-2} + (\bm{I}-\bm{R})^{-1}\right)\bm{1}_1^D.
  \end{split}
\label{eq:eq_avgN}
\end{equation}
Equation \eqref{eq:eq_avgN} shows that we only need $\bm{\pi}_0$, $\bm{\pi}_1$ and $\bm{R}$, thus no need to calculate the infinite number of limiting probabilities.
\begin{theorem}
  A strict upper bound on the average hot data download time of replicate-to-all system with availability one is given as $E[N_{ma}] / \lambda$ where $E[N_{ma}]$ is given in \eqref{eq:eq_avgN}.
%   \begin{equation}
%     E[T] < \frac{E[N_{ma}]}{\lambda}.
%   \label{eq:eq_simplex_t1_matrix_analytic_ub}
%   \end{equation}
\label{thm_reptoall_t1_matrixanalytic_ub}
\end{theorem}
\begin{proof}
 Truncation of the state process is equivalent to imposing a blocking on the recovery group whenever one of the recovery server is ahead of its sibling by two replicas, which works slower than the actual system. Therefore, the average download time found for the truncated system is an upper bound on the actual download time and by Little's law it is expressed as $E[N_{ma}] / \lambda$.
\end{proof}

Fig.~\ref{fig:fig_reptoall_t1_sim_vs_model} shows that the upper bound which we find with matrix analytic procedure is a closer upper bound than the upper bound obtained by the split-merge model in \cite{kadhe2015availability}. This is because the split-merge model is equivalent to truncating the state process and keeping only the central column, while the truncation we consider here keeps five central columns of the state process which yields greater predictive accuracy as expected.

% -------------------------  An approximation using the Conjecture and Renewal Theory for Simplex(t) ------------------------- %
\section{Extension to the General Case}
\label{sec:sec_reptoall_t_conj_approx}
% State space complexity for the join queue under high-traffic approximation exponentially increases with $t$ and analysis of system with multiple recovery groups becomes quickly intractable.
Given the probabilities $f_j$'s over the possible request service time distributions $S_j$'s, we have an $M/G/1$ approximation for the replicate-to-all system with any degree of availability (Prop.~\ref{prop_reptoall_mg1}). However even when the system availability is one ($t=1$), exact analysis is formidable and we could only find estimates for $f_j$'s in Sec.~\ref{sec:sec_reptoall_t1}. In this section, we derive bounds or estimates for $f_j$'s when the system availability is greater than one by building on the relation between $f_j$'s that we conjectured in Conj.~\ref{conj_reptoall_t_fjs}.

\begin{theorem}
  Under the assumptions in Conj.~\ref{conj_reptoall_t_fjs}, we have the following bound and approximations for the probabilities $f_j$'s over the request service time distributions $S_j$'s
 \begin{equation}
    f_0 \geq \frac{1-\rho}{1-\rho^{t+1}}, \quad f_j \approx \rho^j\frac{1-\rho}{1-\rho^{t+1}}, \quad 1 \leq j \leq t.
  \label{eq:eq_simplex_t_bounds_on_f_j}
  \end{equation}
%   probability $f_j$'s over the request service time distribution $S_j$'s are bounded as
%   \begin{equation}
%     f_0 \geq \frac{1-\rho}{1-\rho^{t+1}}, \quad f_j \geq \rho^j \hat{f}_0, \quad 1 \leq j \leq t.
%   \label{eq:eq_simplex_t_bounds_on_f_j}
%   \end{equation}
  for $\rho \geq \max\set{\rho_j}_{j=1}^t$. Approximation for $f_j$'s tends to be an over estimator, especially as $j$ gets larger.
\label{thm_reptoall_t_fsj_by_conj}
\end{theorem}
\begin{proof}
  Recall the relation in Conj.~\ref{conj_reptoall_t_fjs}
  \[ f_{j} = \rho_j f_{j-1}; \quad \rho_j < 1, \quad j = 1, \dots, t. \]
  Using the normalization requirement, we obtain
  \[ \sum_{i=0}^D f_i = f_0 \Bigl[1 + \sum_{j=1}^D \prod_{i=1}^j \rho_i \Bigr] = 1. \]
  Then, $f_j$'s are found in terms of $\rho_j$'s as
  \begin{equation}
      f_0 = \Bigl[1 + \sum_{j=1}^D \prod_{i=1}^j \rho_i \Bigr]^{-1}, \quad f_j = f_0 \prod_{i=1}^j \rho_i.
    \label{eq:eq_fj_intermsof_ro}
  \end{equation}
  Substituting each $\rho_i$ with the same $\rho$, which is an upper bound on all $\rho_i$'s, and solving for $f_j$'s gives the estimates
  \[ \hat{f}_0 = \frac{1-\rho}{1-\rho^{t+1}}, \quad \hat{f}_j = \rho^j f_0 \]
  Note that, the above estimates preserve the relation given in Conj.~\ref{conj_reptoall_t_fjs}. It is easy to see $f_0 \geq \hat{f}_0$. However, concluding that $f_j \leq \hat{f}_j$ for $j \neq 0$ requires further assumptions on $\rho_j$'s but it is also easy to see why $\hat{f}_j$'s tend to be an over estimator for most series of $\rho_j$'s, which is especially the case for larger $j$, e.g., showing that $\hat{f}_t/f_t > 1$ holds is straightforward.
%   \[ \frac{\hat{f}_j}{f_j} \geq \frac{\rho^j \hat{f}_0}{\rho_j \dots \rho_1 f_0} \geq 1, \quad 1 = 0, \dots, t \]
\end{proof}

$M/G/1$ approximation in Prop.~\ref{prop_reptoall_mg1} relies on good estimates for the request service time probabilities $f_j$'s. The tighter the upper bound $\rho$ in Thm~\ref{thm_reptoall_t_fsj_by_conj}, the better the estimates $\hat{f}_j$'s are. The simplest way is to set $\rho$ to its naive maximum, $\rho = 1$, and obtain the estimates as $\hat{f}_j = 1/(t+1)$ for $j = 0, \dots, t$. Substituting these estimates in the $M/G/1$ approximation gives us a \textit{naive approximation} for the replicate-to-all system.
% on average download time, labeled as ``Naive approximation'' in Fig.~\ref{fig:plot_reptoall_t}.

Next, we find an inequality for $\rho$ in Thm.~\ref{thm_reptoall_t_fsj_by_conj}, which leads us to better estimates for $f_j$'s.
\begin{corollary}
  Upper bound $\rho$ in Thm.~\ref{thm_reptoall_t_fsj_by_conj} holds the inequality
  \begin{equation}
    (1 - \lambda E[\hat{S}])\rho^{t+1} - \rho + \lambda E[\hat{S}] \geq 0.
  \label{eq:eq_reptoall_t_ineq_for_ro}
  \end{equation}
  where $E[\hat{S}] = \frac{1}{t+1}\sum_{j=0}^D E[S_j]$.
\label{cor_reptoall_t_ineq_for_ro}
\end{corollary}
\begin{proof}
  Under stability, sub-sequence of request arrivals that find the system empty forms a renewal process \cite[Thm.~5.5.8]{gallager2013stochastic}. Assumptions in Conj.~\ref{conj_reptoall_t_fjs} imply that the replicate-to-all system is an $M/G/1$ queue with arrival rate $\lambda$ and service time $S$. Expected number of request arrivals $E[J]$ between successive renewal epochs (busy periods) is given as $E[J] = 1/(1 - \lambda E[S])$ \cite[Thm.~5.5.10]{gallager2013stochastic}.
  
  Requests that find the system empty at arrival are served with type-$0$ service time distribution. Requests that arrive within a busy period can be served with any type-$j$ service time distribution for $j = 0, \dots, t$. This observation reveals that $1/E[J]$ is a lower bound for $f_0$.
  
  Computing the value of $E[J]$ requires knowing $E[S]$, which we have been trying to estimate. An upper bound is given as $E[J_{ub}] = 1/(1 - \lambda E[S_{lb}])$ where $E[S_{lb}]$ is a lower bound for $E[S]$. One possibility is
  \[ E[S_{lb}] = \frac{1}{t+1}\sum_{j=0}^D E[S_j] \]
  which we previously obtained for the naive $M/G/1$ approximation. Thus, we have
  \[ f_0 \geq 1/E[J] \geq 1/E[J_{ub}] = 1 - \lambda E[S_{lb}]. \]
   
   In the system for which the estimate $\hat{f}_0 = 1/(t+1)$ becomes exact, the lower bound obtained from renewal theory (i.e., $1 - \lambda E[S]$) holds as well under stability. For this system, $E[S_{lb}]$ is exact, hence
   \[ \hat{f}_0 = 1/(t+1) \geq 1 - \lambda E[S_{lb}]. \]
  One can see that $(1-\rho)/(1-\rho^{t+1}) \geq 1/(t+1)$ for $0 \leq \rho < 1$, so we have
  \[ \frac{1-\rho}{1-\rho^{t+1}} \geq \frac{1}{t+1} \geq 1 - \lambda E[S_{lb}] \]
  from which \eqref{eq:eq_reptoall_t_ineq_for_ro} follows.
\end{proof}

Next we use inequality \eqref{eq:eq_reptoall_t_ineq_for_ro} to get a tighter value for the upper bound $\rho$ in Thm.~\ref{thm_reptoall_t_fsj_by_conj} as follows. Solving for $\rho$ in \eqref{eq:eq_reptoall_t_ineq_for_ro} does not yield a closed form solution, so to get one we take the limit as
\[ \lim_{t \to \infty} (1 - \lambda E[\hat{S}])\rho^{t+1} - \rho + \lambda E[\hat{S}] \geq 0 \iff \rho \leq \lambda E[\hat{S}]. \]

Using this new upper bound instead of the naive one (i.e., $\rho = 1$) gives us a \textit{better} $M/G/1$ approximation for replicate-to-all system. We next present the \textit{best} $M/G/1$ approximation we could obtain by estimating each $\rho_j$ separately rather than using a single upper bound $\rho$.

\begin{corollary}
   In replicate-to-all system, the service time probabilities $f_j$'s are well approximated as
  \begin{equation}
    \hat{f}_0 = \Bigl[1 + \sum_{i=1}^D \prod_{j=0}^{i-1} \hat{\rho}_j\Bigr]^{-1}, \quad \hat{f}_j = \hat{f}_0 \prod_{i=0}^{j-1} \hat{\rho}_i
  \end{equation}
  where $\hat{\rho}_j$'s are computed recursively as
  \begin{equation}
  \begin{split}
    \hat{\rho}_0 &= \frac{\lambda E[\hat{S}]}{t(1 - \lambda E[\hat{S}])}, \\
    \hat{\rho}_j &= \frac{1 - (1 - \lambda E[\hat{S}])(1 + \sum_{k=0}^{j-1} \prod_{l=0}^{k} \hat{\rho}_l)}{(1 - \lambda E[\hat{S}])(t-j)\prod_{k=0}^{j-1} \hat{\rho}_k}, \quad j = 1, \dots, t.
  \end{split}
  \label{eq:eq_reptoall_t_incremental_ros}
  \end{equation}
  where $E[\hat{S}] = \frac{1}{t+1}\sum_{j=0}^D E[S_j]$.
\label{cor_reptoall_t_approx_w_renewaltheo}
\end{corollary}
\begin{proof}
  Setting $\hat{f}_j = \hat{\rho}_0\hat{f}_0$ for $j = 1, \dots, t$ and using the normalization requirement $\sum_{i=0}^D \hat{f}_i = 1$, we find $\hat{f}_0 = 1/(1+t\hat{\rho}_0)$. Using the inequality $1/(1+t\hat{\rho}_0) \geq 1/(t+1) \geq 1 - \lambda E[\hat{S}]$ (as found in the proof of Cor~\ref{cor_reptoall_t_ineq_for_ro}), we get the upper bound
  \[ \hat{\rho}_0 \leq \frac{\lambda E[\hat{S}]}{1 - \lambda E[\hat{S}]}. \]
  
  Fixing the value of $\hat{\rho}_0$ to the upper bound above and setting $\hat{f}_1 = \hat{\rho}_0\hat{f}_0$, $\hat{f}_i = \hat{\rho}_1\hat{\rho}_0\hat{f}_0$ for $2 \leq i \leq t$, we find an upper bound on $\hat{\rho}_1$ by executing the same steps that we took for finding the upper bound on $\hat{\rho}_0$. Normalization requirement gives $\hat{f}_0 = \bigl[1+\hat{\rho}_0(1+\hat{\rho}_1(t-1))\bigr]^{-1}$ and we have $\hat{f}_0 \geq 1/(1+t) \geq 1 - \lambda E[\hat{S}]$, which yields
  \[ \hat{\rho}_1 \leq \frac{1-(1 - \lambda E[\hat{S}])(1-\hat{\rho}_0)}{(t-1)(1 - \lambda E[\hat{S}])\hat{\rho}_0}. \]

  The same process can be repeated to find an upper bound on $\hat{\rho}_2$ by fixing $\hat{\rho}_0$ and $\hat{\rho}_1$ to the upper bounds above. Generalizing this, fixing $\hat{\rho}_0$, ..., $\hat{\rho}_{j-1}$ to their respective upper bounds, we find the upper bound
  \[ \hat{\rho}_j \leq \frac{1 - (1 - \lambda E[\hat{S}])\bigl(1 + \sum_{k=0}^{j-1} \prod_{l=0}^{k} \hat{\rho}_l\bigr)}{(1 - \lambda E[\hat{S}])(t-j)\prod_{k=0}^{j-1} \hat{\rho}_k} \]
  
  Finally, setting each $\hat{\rho}_j$ to its respective upper bound allows us to compute the estimates $\hat{f}_i$'s.
\end{proof}

To summarize, we obtained a naive $M/G/1$ approximation by setting $\rho = 1$ in Thm.~\ref{thm_reptoall_t_fsj_by_conj}. Next, using Cor.~\ref{cor_reptoall_t_ineq_for_ro} we obtained a tighter bound on $\rho$, which gave the better approximation. Finally, in Cor.~\ref{cor_reptoall_t_approx_w_renewaltheo} we find the best approximation that we could by computing the estimates for $\rho_j$'s recursively. Fig.~\ref{fig:plot_reptoall_t} gives a comparison of these naive, better and the best approximations with the simulated average hot data download time. Note that the approximations we derived in this section give close approximations system with availability one, for which we were able to obtain a very good approximation using the high traffic assumption in Sec.~\ref{sec:sec_reptoall_t1}.
\begin{figure}[htbp]
  \centering
  \begin{subfigure}[h]{.4\textwidth}
    \centering
    \includegraphics[width=1\textwidth, keepaspectratio=true]{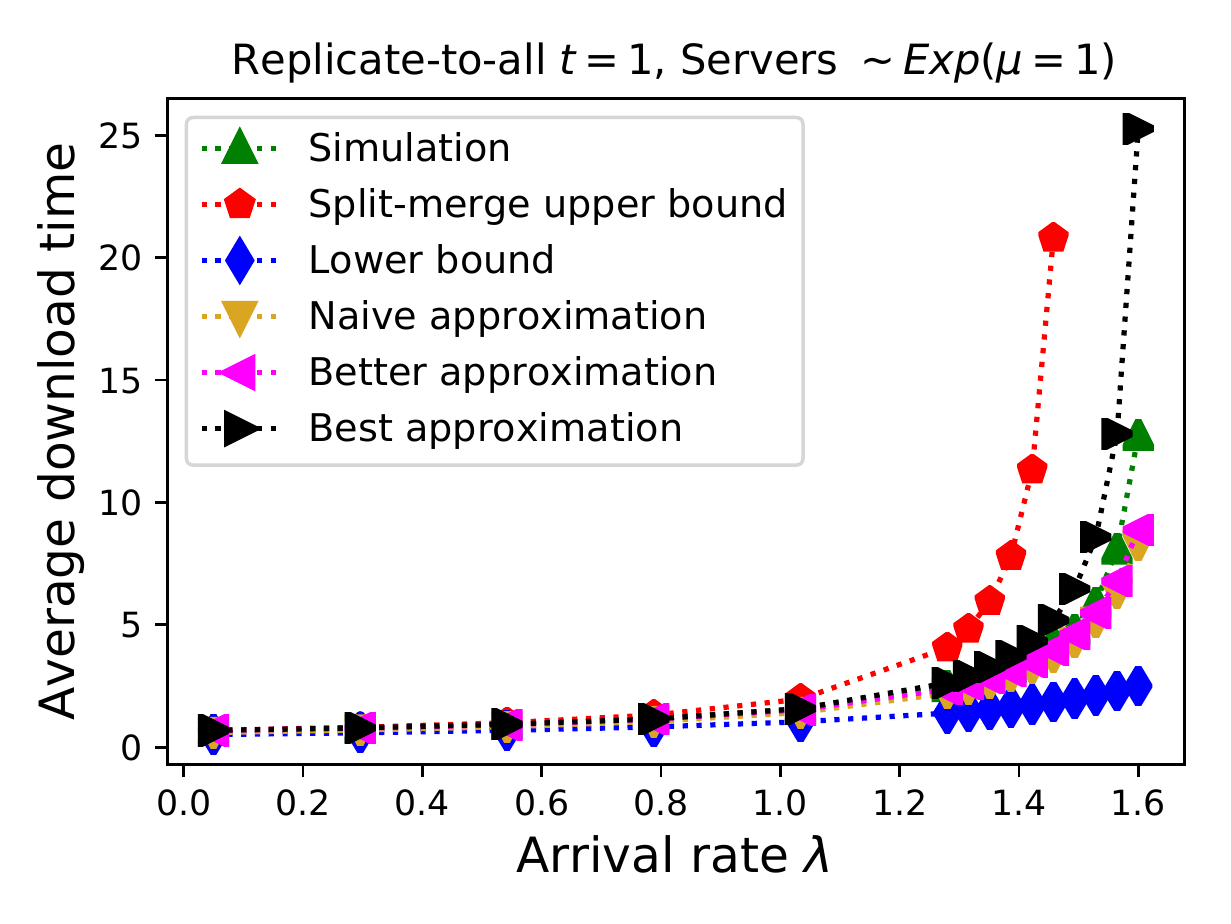}
  \end{subfigure}
  \begin{subfigure}[h]{.4\textwidth}
    \centering
    \includegraphics[width=1\textwidth, keepaspectratio=true]{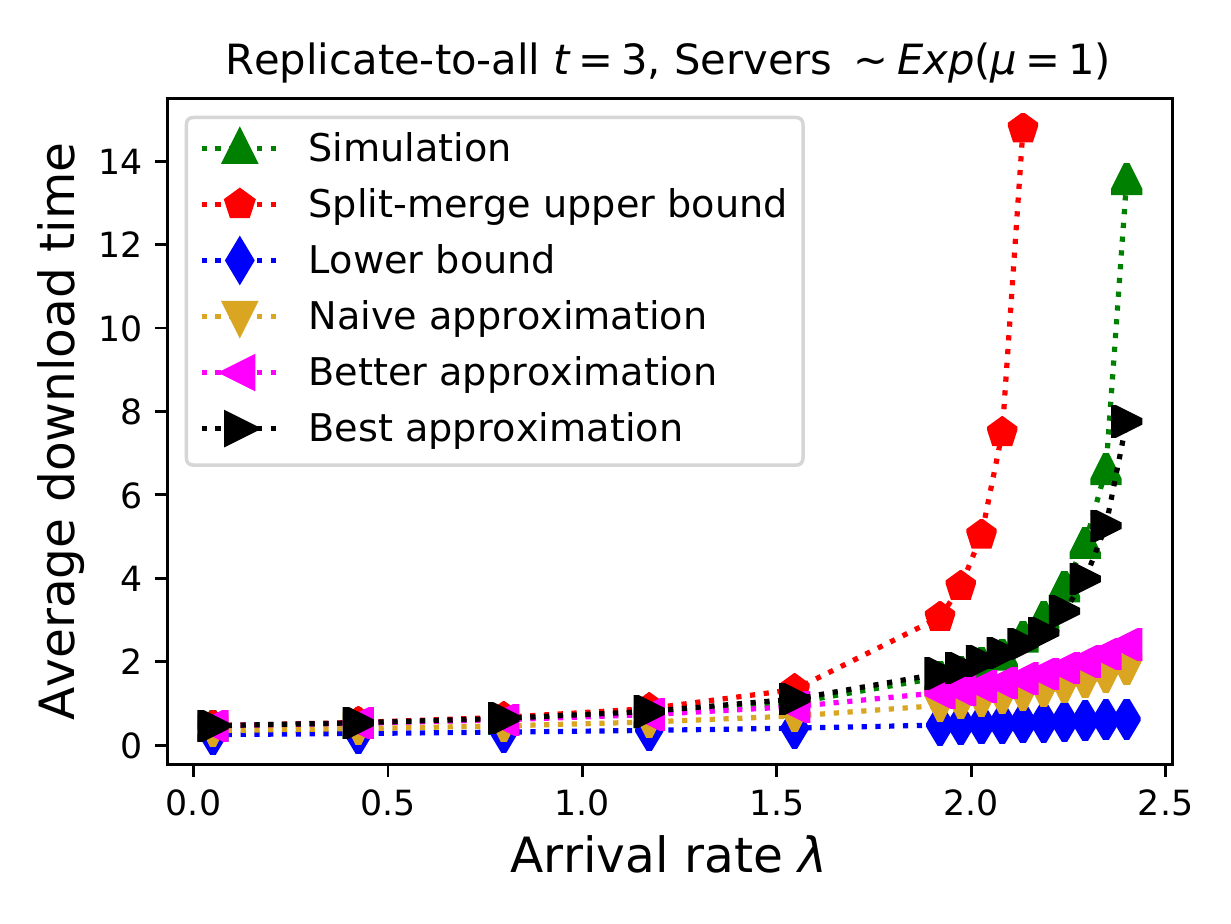}
  \end{subfigure}
  \caption{Comparison of the approximations given in Sec.~\ref{sec:sec_reptoall_t1} and the simulated average hot data download time for replicate-to-all system with availability $t=1$ (Left) and $t=3$ (Right).}
  \label{fig:plot_reptoall_t}
\end{figure}

% In Appendix \ref{subsec:subsec_simplex_t_lb_varki_gauri}, another lower bound is found by analyzing an equivalent model for Simplex($t$). It is shown as $E[\hat{T}_{fast-serial}]$ in Fig.~\ref{fig:plot_reptoall_t} and is a much looser lower bound than $E[\hat{T}(1)]$ and $E[\hat{T}(\hat{\rho})]$, especially for higher job arrival rate.

% ---------------------------------------  Non-Exp servers  --------------------------------- %
\subsection{Non-exponential servers}
\label{subsec:subsec_reptoall_t_nonexp_servers}
Modern distributed storage systems have been shown to be susceptible to heavy tails in service time \cite{xu2013bobtail, TailAtScale:DeanB13, GoogleClusterData:ReissTG12}. Here we evaluate the $M/G/1$ approximation presented in Prop.~\ref{prop_reptoall_mg1} and Cor.~\ref{cor_reptoall_t_approx_w_renewaltheo} on server service times that are advocated in the literature to better model the server side variability in practice.

We consider two analytically tractable service time distributions. First one is $\textit{Pareto}(s, \alpha)$ with minimum value $s > 0$ and tail index $\alpha > 1$. Pareto is a canonical heavy tailed distribution (smaller $\alpha$ means heavier tail) that is observed to fit service times in real systems \cite{TailAtScale:DeanB13}. Second one is $\textit{Bernoulli}(U, L, p)$ that takes value $L$ (long i.e., $L > U$) w.p. $p < 0.5$ and value $U$ (usual) otherwise, which is also mentioned to be a proper model for the server side variability in practice \cite{TailAtScale:DeanB13}.

Two main observations about the replicate-to-all system that lead us to the $M/G/1$ approximation hold also when the server service times are non-exponential; requests depart in the order they arrive and the system experiences frequent time epochs that break the chain of dependence between the request service times. However, when the memoryless property of exponential distribution is missing, there are infinitely many request service time distributions possible. As discussed in Sec.~\ref{subsec:subsec_reptoall_mg1_approx}, the set of possible request service time distributions can be trimmed down to size $t+1$ by modifying the system and forcing a service restart for the request copies that start service earlier as soon as their associated request move to head of the line (i.e., once all the copies of a request moves in service).

% Repeating the description given in Lemma~\ref{obv_reptoall} for Non-exponential servers, type-$j$ request service time distribution can be written for $j = 0, \dots, t$ as
% \begin{equation*}
%   S_j \sim \min\{V_0, \ldots, V_j, V^1_{2:1}, \dots, V^{t-j}_{2:1}\}.
% \label{eq:eq_typei_servtime}
% \end{equation*}
% where $V_i$'s are i.i.d. server service times, with CDF $F(v)$. It is easy to derive the tail of $S_j$ as
% \begin{equation}
%   Pr\{S_j > s\} = (1 - F(s))^{t+1}(1 + F(s))^{t-j}.
% \label{eq:eq_typei_servtime_tail}
% \end{equation}
% Notice that, tail of request service time is smaller for greater values of $i$, in other words, types of request services can be ordered from slow to fast as $V_0 > \ldots > V_t$.

Repeating the formulation in Lm.~\ref{lm_reqservtime_dists} for \textit{Pareto} server service times, moments for type-$j$ request service time distribution is found for $j = 0, \dots, t$ as
\begin{equation}
\begin{split}
%   E[V_i] &= \sum\limits_{k=0}^{t-i} {{t-i}\choose{k}} 2^k (-1)^{t-i-k} \lambda(1 + \frac{1}{\alpha(2t+1-i-k)-1}), \\
%   E[V_i^2] &= \sum\limits_{k=0}^{t-i} {{t-i}\choose{k}} 2^k (-1)^{t-i-k} \lambda^2(1 + \frac{1}{\frac{\alpha}{2}(2t+1-i-k)-1}).
  E[S_j] &= \sum\limits_{k=0}^{t-j} {{t-j}\choose{k}} 2^k (-1)^{t-j-k} \lambda\frac{\alpha(2t+1-j-k)}{\alpha(2t+1-j-k) - 1}, \\
  E[S_j^2] &= \sum\limits_{k=0}^{t-j} {{t-j}\choose{k}} 2^k (-1)^{t-j-k} \lambda^2\frac{\alpha(2t+1-j-k)}{\alpha(2t+1-j-k) - 2}.
  \end{split}
  \label{eq:eq_reqservmoments_Pareto_servers}
\end{equation}
and when server service times are \textit{Bernoulli}, we find
\begin{equation}
\begin{split}
  E[S_j] &= U + (L-U)p^{t+1}(2-p)^{t-j}, \\
  E[S_j^2] &= U^2 + (L^2-U^2)p^{t+1}(2-p)^{t-j}.
  \end{split}
  \label{eq:eq_reqservmoments_Bern_servers}
\end{equation}

% Service time moments for an arbitrary request can then be approximated as the weighted sum of $E[V_i]$'s for $0 \leq i \leq t$, where weights are the service time probabilities $f_i$'s as given in Proposition~\ref{prop_reptoall_mg1}. Same ideas discussed in Subsection~\ref{subsec:subsec_simplex_t_conjecture} apply for Non-exponential servers and service time probabilities can then be estimated using the recursive method given in Cor.~\ref{cor_reptoall_t_approx_w_renewaltheo}.

Moments given above enable the $M/G/1$ approximation stated in Prop.~\ref{prop_reptoall_mg1} and Cor.~\ref{cor_reptoall_t_approx_w_renewaltheo}.
Fig.~\ref{fig:plot_reptoall_Gserv_sim_vs_approx} illustrates that the approximation yields a close estimate for the average hot data download time when service times at the servers are distributed as \textit{Pareto}~\footnote{Simulating queues with heavy tailed service time is difficult and numeric results given here are optimistic i.e., plotted values serve as a lower bound \cite{yucesan2002difficulties}.} or \textit{Bernoulli}. When the tail heaviness pronounced by the server service times is high (e.g., \textit{Pareto} distribution), system operates more like the restricted split-merge model. This is because under heavier tail it is harder for the leading servers in the recovery groups to keep leading, and it is more likely for a leading server to straggle and loose the lead. This effect is greater on the leading servers as the system availability increases since the leading servers compete with more servers to keep leading. This increases the fraction of requests that have type-$0$ (slowest) service time, hence split-merge bound gets tighter at higher availability.
% Recall that the observation that service times are independent across these renewal epochs was crucial for the M/G/1 approximation and here the increase in the frequency of these renewals makes the approximation work quite well. In addition, it increases the fraction of requests that have slower service time, thus the response time is closer to that of split-merge system.
\begin{figure}[htbp]
  \centering
  \begin{subfigure}[h]{.4\textwidth}
    \centering
    \includegraphics[width=1\textwidth, keepaspectratio=true]{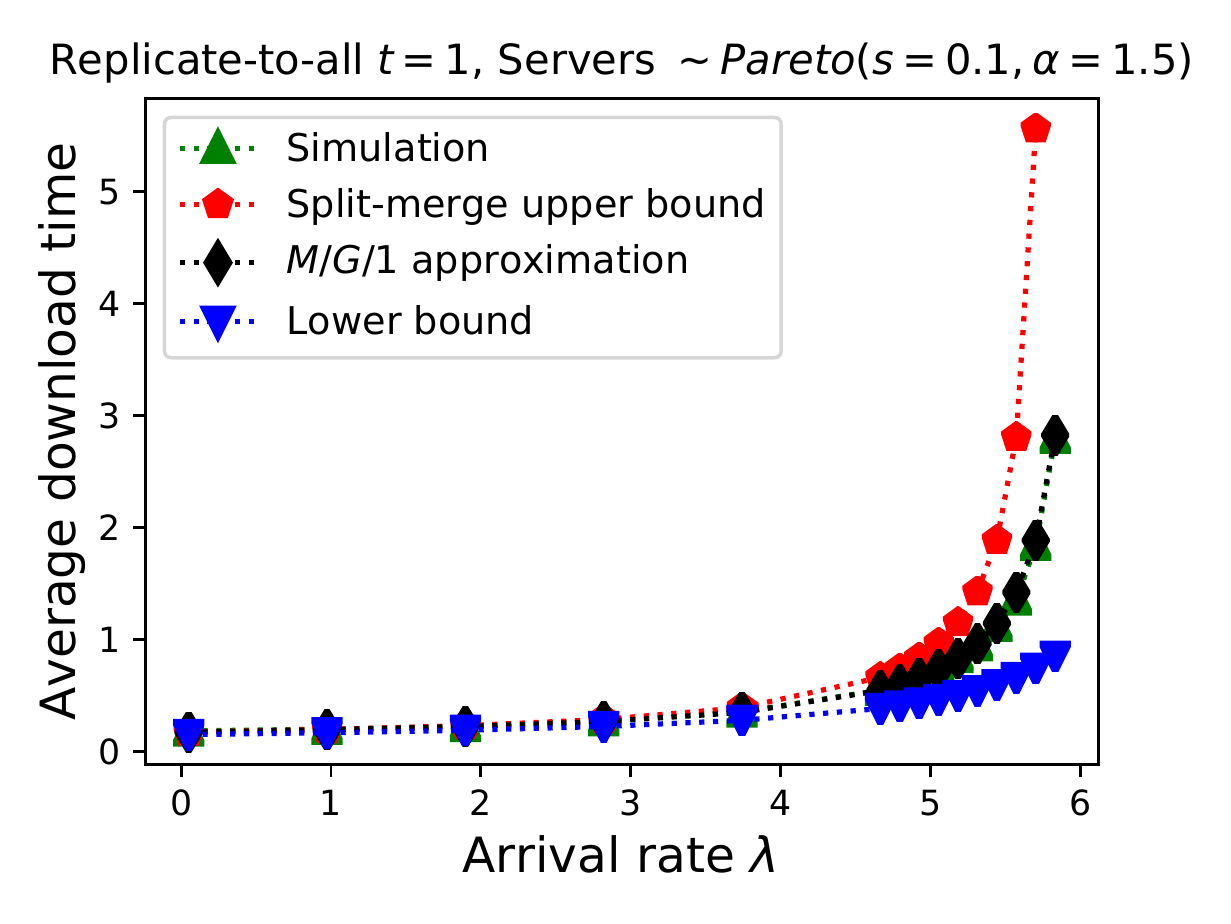}
  \end{subfigure}
  \begin{subfigure}[h]{.4\textwidth}
    \centering
    \includegraphics[width=1\textwidth, keepaspectratio=true]{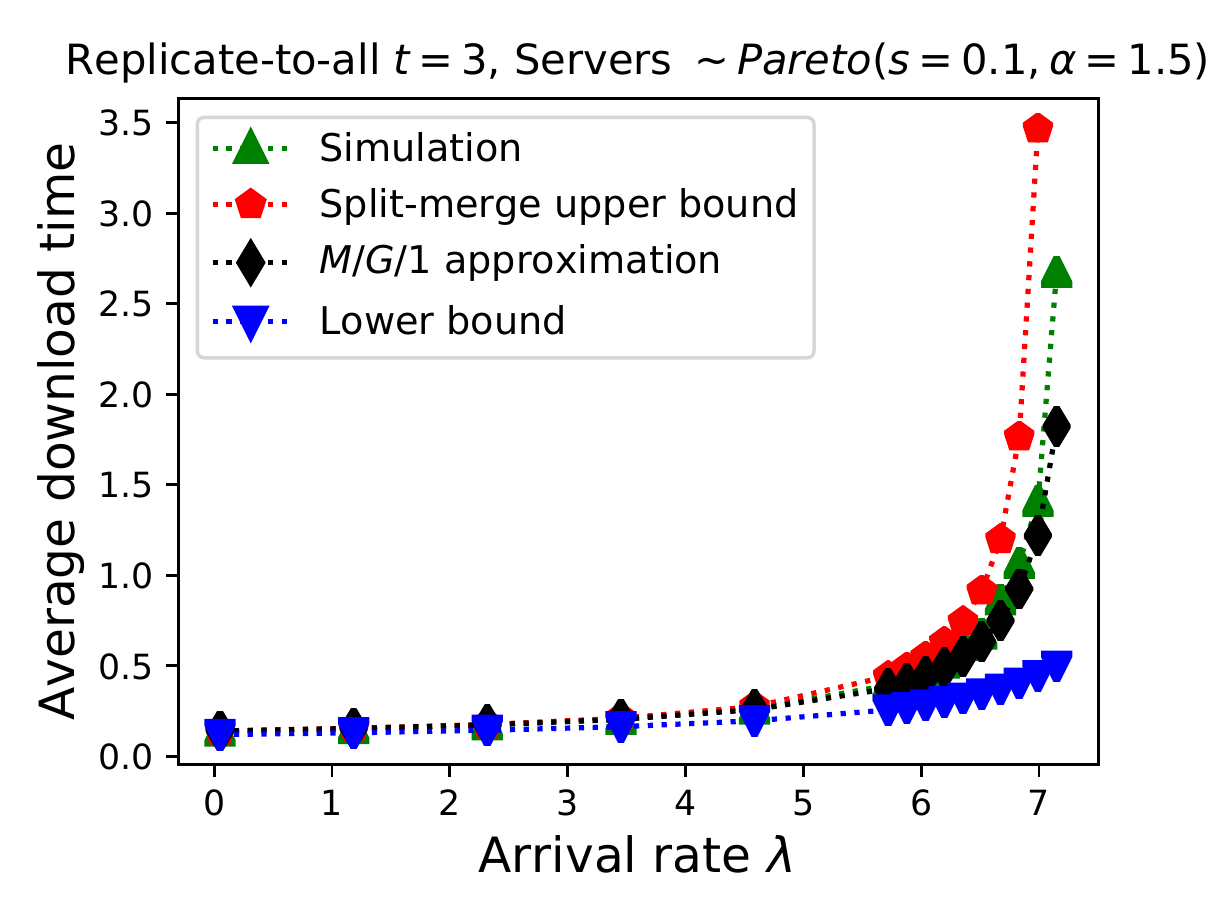}
  \end{subfigure}
  \begin{subfigure}[h]{.4\textwidth}
    \centering
    \includegraphics[width=1\textwidth, keepaspectratio=true]{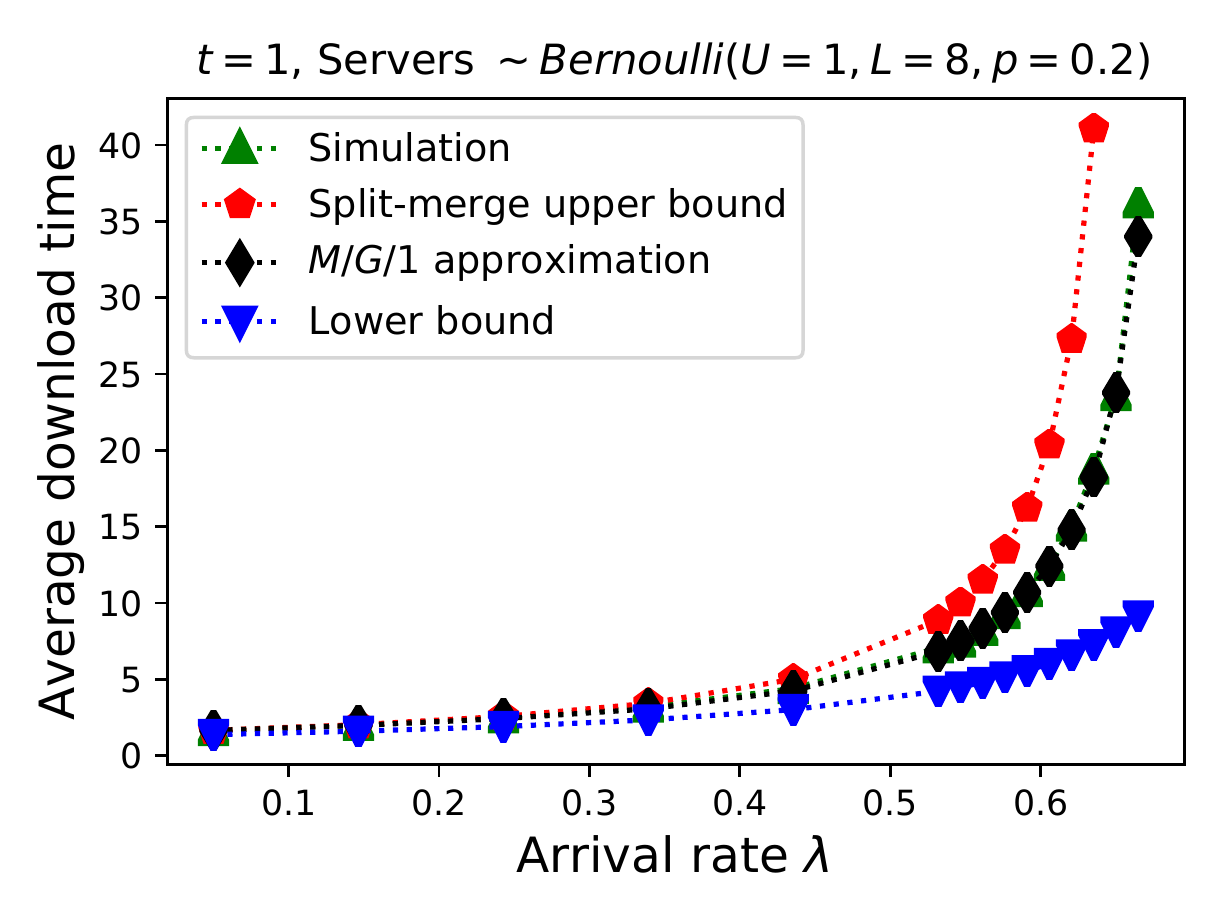}
  \end{subfigure}
  \begin{subfigure}[h]{.4\textwidth}
    \centering
    \includegraphics[width=1\textwidth, keepaspectratio=true]{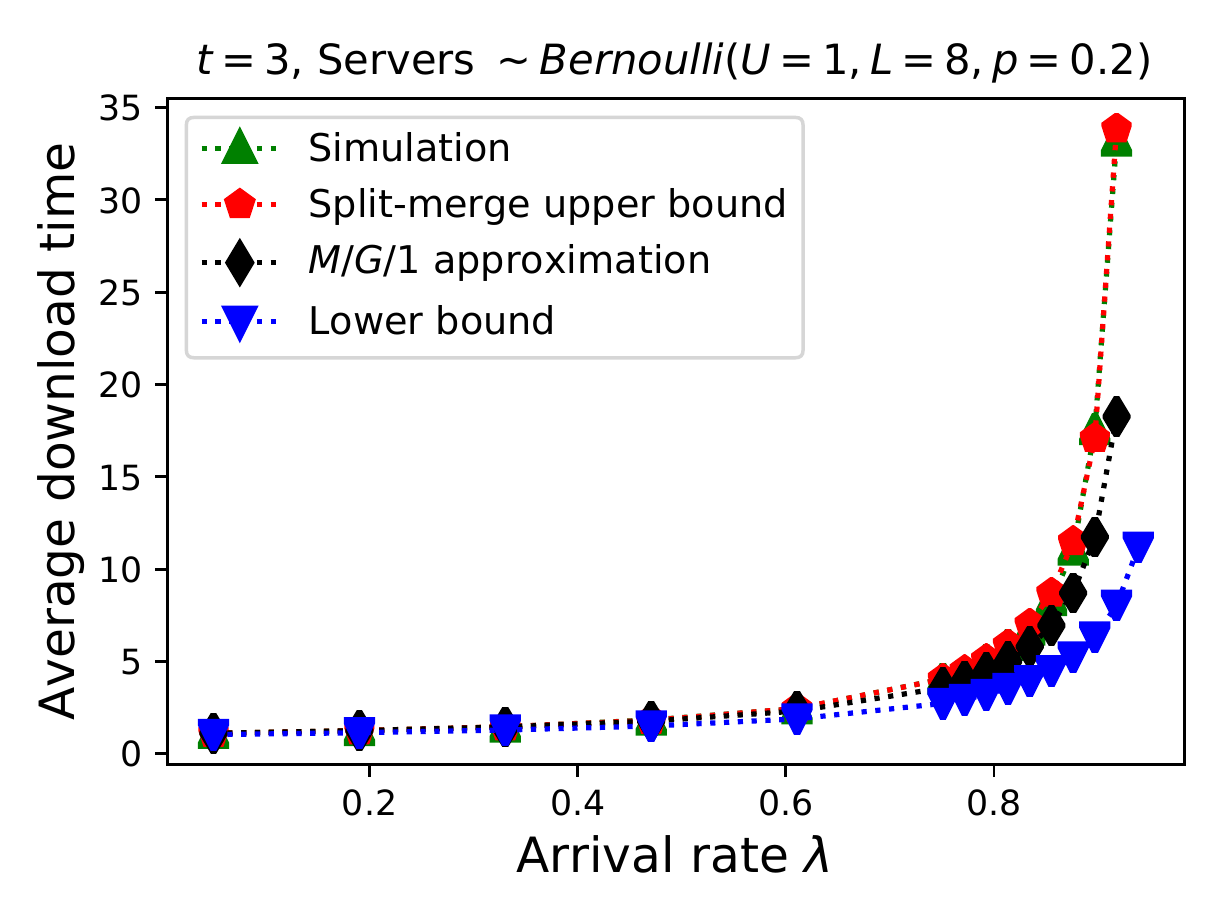}
  \end{subfigure}
  \caption{Comparison of the $M/G/1$ approximation in Cor.~\ref{cor_reptoall_t_approx_w_renewaltheo}, the split-merge upper bound in \cite{kadhe2015availability}, the lower bound given in Thm.~\ref{thm_reptoall_t_lb}, and the simulated average hot data download time for replicate-to-all system under \textit{Pareto} (Top) and \textit{Bernoulli} (Bottom) service times at the servers.}
  \label{fig:plot_reptoall_Gserv_sim_vs_approx}
\end{figure}

% ---------------------------------------  Rep-to-all Mixed-arrivals   --------------------------------- %
\subsection{Hot data download under mixed arrivals}
\label{subsec:subsec_reptoall_t_mixed_arrival}
Recall that in Sec.~\ref{sec:sec_sys_model}, we fixed the system model on the \textit{fixed arrival} scenario where the requests arriving in a busy period ask for only one data symbol, which we refer to as hot data. This was based on the fact that cold data in practice is very rarely requested \cite{CopingWithSkewedContentPopularityInMapreduce:AnanthanarayananAK11}.
We here try to relate the discussion we had so far on the replicate-to-all system to the \textit{mixed arrival} scenario.

Under fixed arrival scenario, the role of each server in terms of storing a systematic or a coded data symbol is fixed for each arriving request. Under mixed arrival scenario, the role of each server depends on the data symbol requested by an arrival. If the requested symbol is $a$ then the systematic server for the request is the one that stores $a$ while the others are recovery servers. When the requested data symbol changes, the role of every server changes.

Under mixed arrival scenario, multiple requests can be simultaneously served at their corresponding systematic servers, hence requests do not necessarily depart in the order they arrive and multiple requests can depart simultaneously. Analysis presented in the previous sections for the replicate-to-all system under fixed arrival scenario mainly depended on the observations that lead us to approximate the system as an $M/G/1$ queue, while the system under mixed arrival scenario surely violates these observations. Therefore, the analysis of the system under mixed arrival scenario turns out to be much harder. However, we can compare the performance of the system under these two arrival scenarios as follows.
\begin{theorem}
  In replicate-to-all system, hot data download time under mixed arrivals is a lower bound for the case with fixed arrivals.
\label{thm_mixed_arr_lb_fixed_arr}
\end{theorem}
\begin{proof}
  Under fixed arrivals, one of the servers in each recovery group may be ahead of its sibling at any time and proceed with a copy of a request that is waiting in line. Therefore, departure of the request copies at the leading servers cannot complete a request alone and early departing copies can only shorten the service duration of a request (recall that type-$j$ service time is stochastically less than type-$(j-1)$, see Lm.~\ref{lm_reqservtime_dists}). However under mixed arrivals, a leading server may be the systematic server for a request waiting in line and completion of its copy at the leading server can complete the request. This is a bigger reduction in the overall download time than completion of a mere copy that fasten the service time of a request, which concludes the proof.
\end{proof}
Fig.~\ref{fig:plot_reptoall_t_mixed_traff} illustrates Thm.~\ref{thm_mixed_arr_lb_fixed_arr} by comparing the simulated values of average hot data download time under fixed and mixed arrival scenarios.
\begin{figure}[H]
  \centering
  \begin{subfigure}[h]{.4\textwidth}
    \centering
    \includegraphics[width=1\textwidth, keepaspectratio=true]{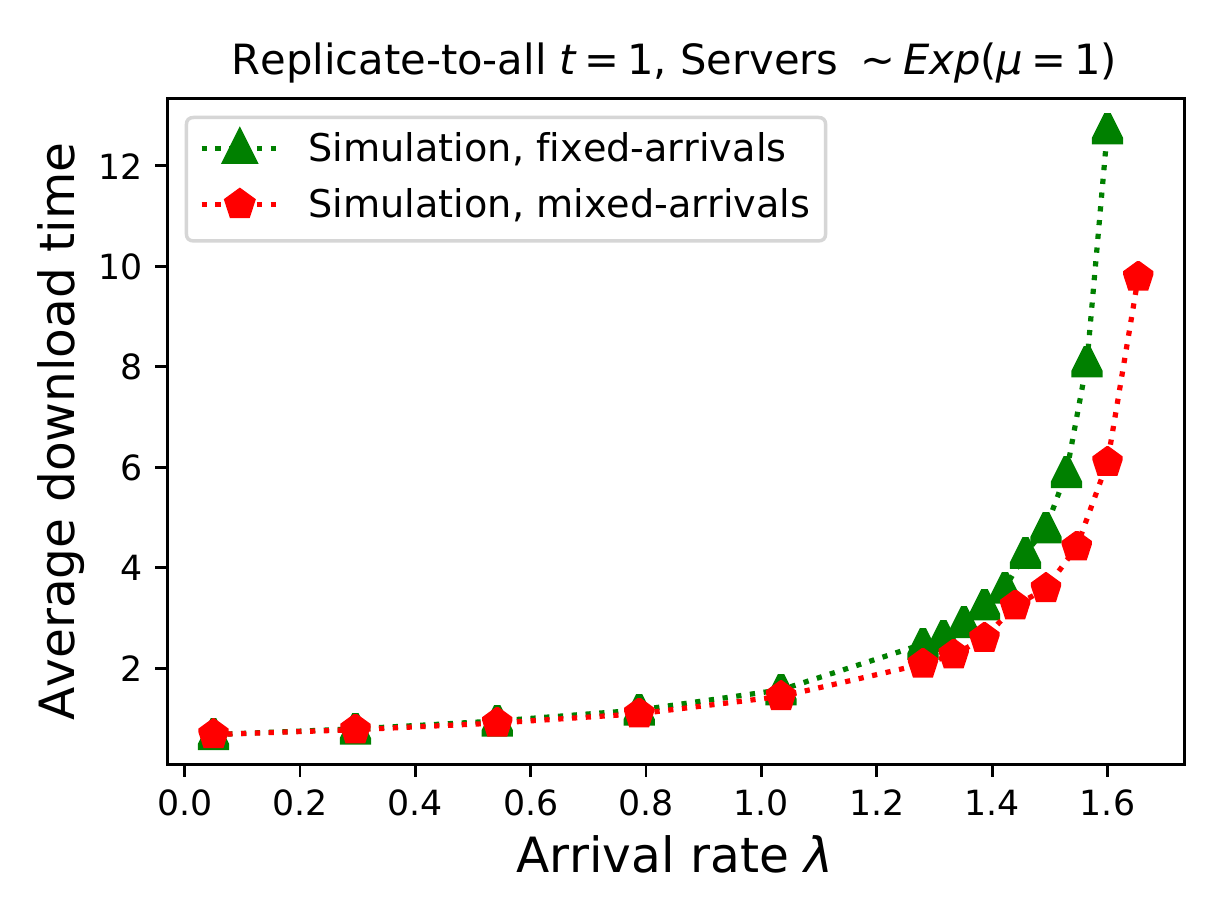}
  \end{subfigure}
  \begin{subfigure}[h]{.4\textwidth}
    \centering
    \includegraphics[width=1\textwidth, keepaspectratio=true]{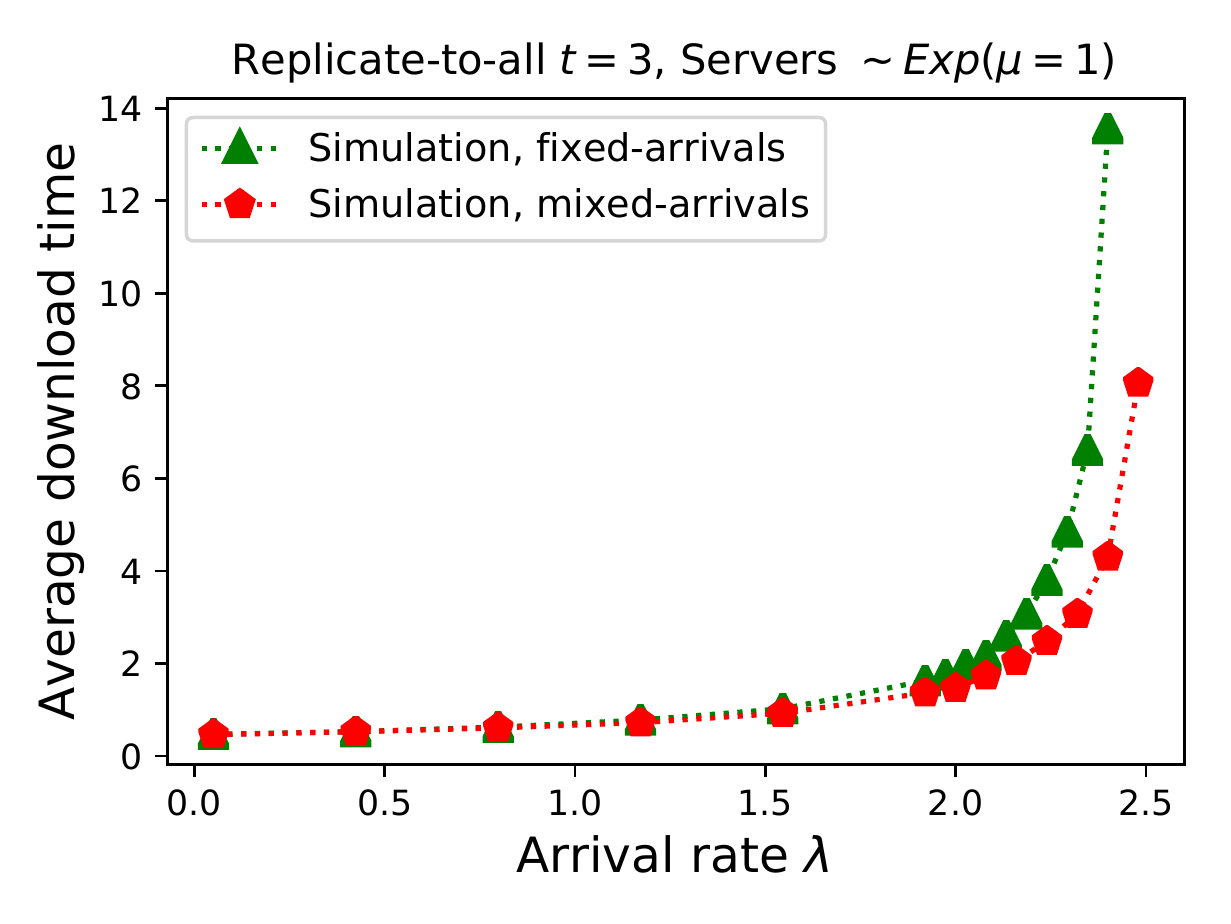}
  \end{subfigure}
  \caption{Comparison of the simulated average hot data download time in replicate-to-all system under fixed (i.e., all requests arriving within a busy period ask for hot data) and mixed arrival scenarios (i.e., arriving requests may ask for any one of the stored data symbols equally likely).}
  \label{fig:plot_reptoall_t_mixed_traff}
\end{figure}

% ############################################  Select-one  ######################################### %
\section{Select-one Scheduling}
\label{sec:sec_selectone}
Storage systems experience varying download load that exhibits certain traffic patterns throughout the day \cite{GFS:GhemawatGL03, WorkloadModelingBook:Feitelson15}. Phase changes in the system load can be anticipated or detected online, which allows for adapting the request scheduling strategy to enable faster data access.
In a distributed setting, data access schemes, hence the request scheduling strategies are either imposed or limited by the availability of data across the storage nodes. The desired scheduling strategies that are optimized for different load phases may require a change in the data availability layout. For instance, hot data may need to be replicated to achieve stable download performance, or replicas for the previously hot but currently cold data may need to be removed in order to open up space in the storage, e.g., see \cite{SurveyOfDynamicRepStrategies:Amjad12} for a survey of adaptive replica management strategies. In a coded storage, data availability can be changed by fetching, re-encoding and re-distributing the stored data. This online modification of data availability puts additional load on the system and introduces additional operational complexity that makes the system prone to operational errors \cite{GFS:GhemawatGL03}.

% % Therefore, over the data encoded with the simplex code, load balancing strategy can also be implemented besides fast data access time with redundant requests.
LRCs, in particular simplex codes, provide the necessary flexibility for the system to seamlessly switch between different request scheduling strategies. Batch codes have been proposed for load balancing purposes in \cite{BatchCodes:IshaiKO04} and their connection to LRCs is studied in \cite{BatchVsLRCs:Skachek18}. A simplex code is a linear batch code. In a simplex coded storage system, the simplest strategy for load balancing, namely {\it select-one}, is to assign an arriving request either to the systematic server or to one of the recovery groups chosen at random according to a scheduling distribution.
Availability of a simplex code allows replicate-to-all and select-one strategies to be used interchangeably. For low or middle traffic regime, replicate-to-all achieves faster data download, however system operates under stability over a greater range of request arrival rate with select-one strategy (see Fig.~\ref{fig:plot_reptoall_t3_vs_selectone}). This implements the capability for the system to switch between the replicate-to-all and select-one schedulers seamlessly depending on the measured system load.

If we assume that arrivals for cold data are negligible (i.e., fixed arrival scenario as defined in Sec.~\ref{subsec:subsec_reptoall_t_mixed_arrival}), select-one scheduler splits the arrival stream of hot data requests into independent Poisson streams flowing to the systematic server and to each recovery group. Then, the systematic server implements an $M/G/1$ queue while each recovery group implements an independent fork-join queue with two servers and Poisson arrivals. Download time for an arbitrary request can then be found as the weighted sum of the response times of each of the sub-systems where the weights are the scheduling probabilities across the sub-systems.

When the service times at the servers are exponentially distributed, we state below an exact expression for the average hot data download time using an exact expression given in \cite{flatto1984two} for the average response time of a two-server fork-join queue.
In general, an approximation on the average download time can be found along the same lines using the approximations available in the literature on the response time of fork-join queues with an arbitrary service time distribution. We refer the reader to \cite{SurveyOnFJQs:Thomasian15} for an excellent survey of the available approximations for fork-join queues.
\begin{theorem}
  Suppose that service times at the servers are exponentially distributed with rate $\mu$.
  Given a scheduling distribution $[p_0, p_1, \dots, p_t]$ respectively over the systematic server and $t$ recovery groups, average hot data download time in the select-one system is given as
  \begin{equation}
    E[T] = \frac{p_0}{\mu-p_0\lambda} + \sum_{i=1}^{t} p_i\frac{12\mu-p_i\lambda}{8\mu(\mu-p_i\lambda)}
  \label{eq:eq_simplex_t_exact_for_select_one}
  \end{equation}
\label{thm_selectone}
\end{theorem}
\begin{proof}
  Each arriving request is independently assigned either to the systematic server with probability $p_0$ or to recovery group-$i$ with probability $p_i$ for $i = 1, \dots, t$. Given a $\textit{Poisson}(\lambda)$ hot data request arrival stream, arrivals that flow to the systematic server form an independent $\textit{Poisson}(p_0\lambda)$ stream while those that flow to repair group-$i$ form an independent $\textit{Poisson}(p_i\lambda)$ stream. Then the systematic server is an $M/M/1$ queue and each repair group is a fork-join queue with two servers and Poisson arrivals, for which an exact expression is given in \cite{nelson1988approximate} for its average response time.
\end{proof}

\begin{figure}[H]
  \centering
  \begin{subfigure}[h]{.4\textwidth}
    \centering
    \includegraphics[width=1\textwidth, keepaspectratio=true]{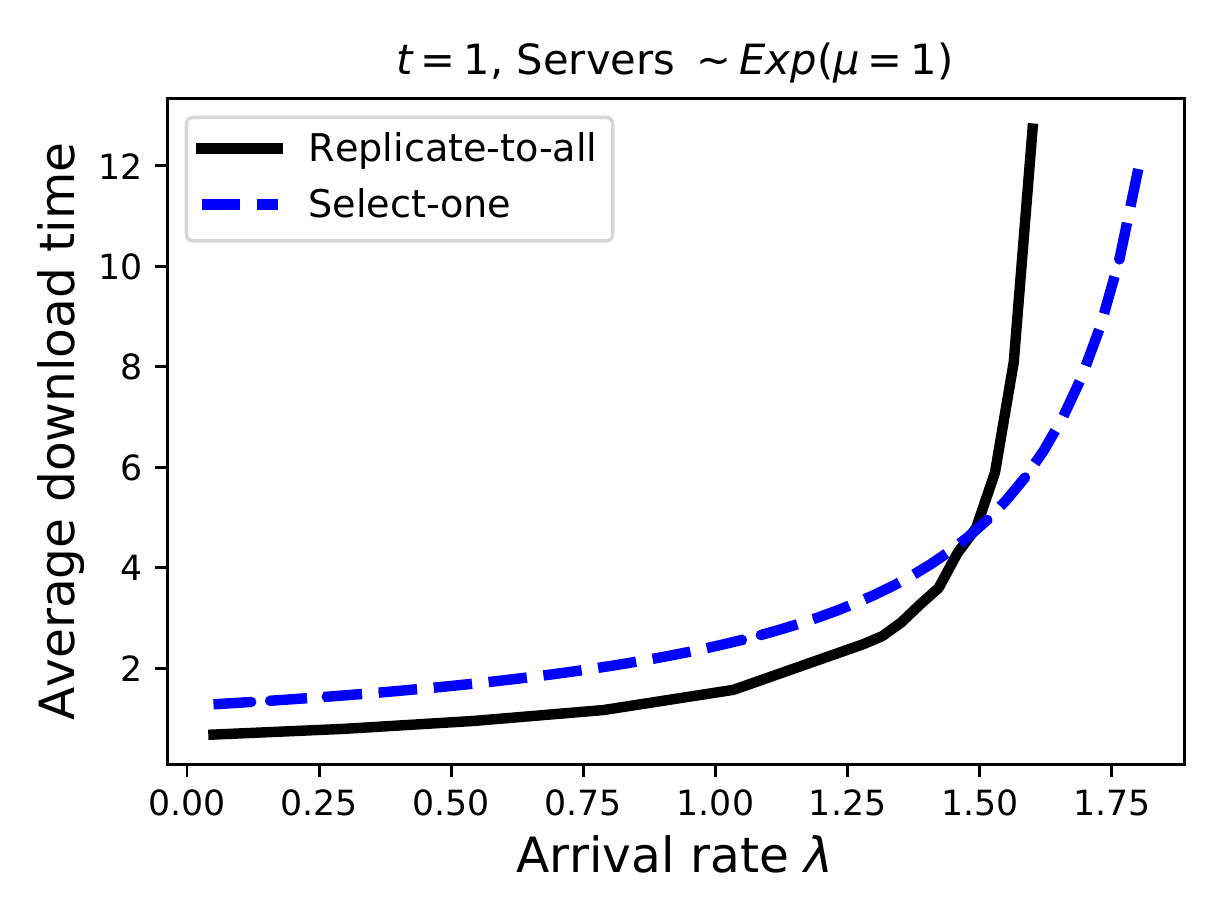}
  \end{subfigure}
  \begin{subfigure}[h]{.4\textwidth}
    \centering
    \includegraphics[width=1\textwidth, keepaspectratio=true]{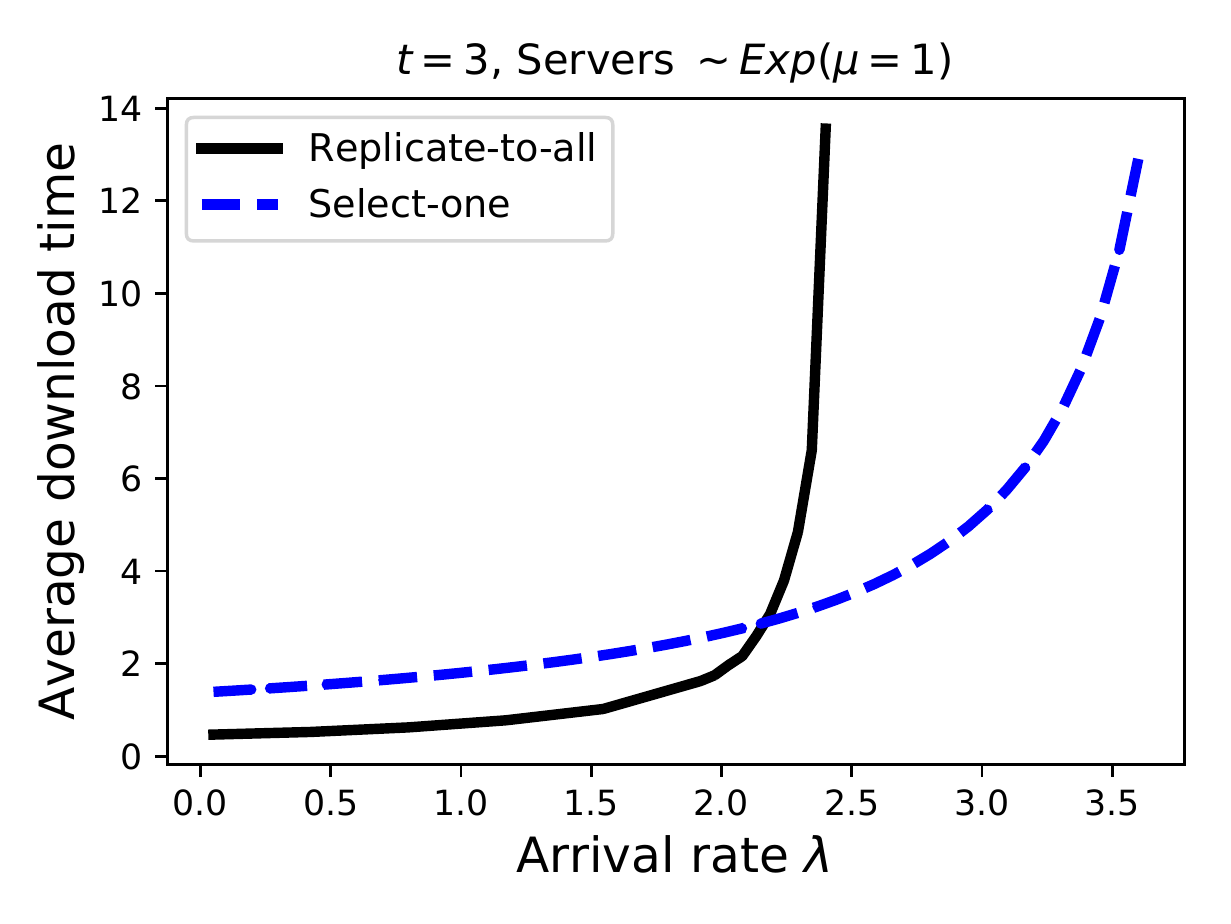}
  \end{subfigure}
  \caption{Comparison of the average hot data download time in replicate-to-all and select-one systems.}
%   Split-to-one here uses the best scheduling distribution possible for each arrival rate.
  \label{fig:plot_reptoall_t3_vs_selectone}
\end{figure}

% ############################################  Fairness-First  ######################################### %
\section{Fairness-First Scheduling}
\label{sec:sec_fairnessfirst}
In a simplex coded storage, we assume that hot and cold data are coded and stored together. Replicate-to-all scheduling that we studied in Sec.~\ref{sec:sec_reptoall} and \ref{sec:sec_reptoall_t1} aims to exploit download with redundancy for all the arriving hot data requests. Our study of replicate-to-all system is centered around the assumption that the traffic for cold data requests is negligible as observed in practice e.g., only 10\% of the stored content is frequently and simultaneously accessed~\cite{CopingWithSkewedContentPopularityInMapreduce:AnanthanarayananAK11}. In this section, we consider the case when the request traffic for cold data is non-negligible.

In replicate-to-all system, hot data requests are replicated to all servers at arrival, including the servers that store cold data. For instance in a $[a, b, a+b]$-system, requests for hot data $a$ are replicated both to its systematic server-$a$ and the server-$b$ that hosts the cold data. Then, the arriving requests for $b$ may end up waiting for the hot data request copies at server-$b$. This increases the cold data download time and is unfair to the cold data requests. In \cite{RedForFairness:Gardner17}, authors propose a fair scheduler for systems with replicated redundancy, where the original requests are designated as primary and assigned with high service priority at the servers while the redundant request copies are designated as secondary and assigned lower service priority.

% Explain FF
We here consider the \textit{fairness-first} scheduler for implementing fairness while opportunistically exploiting the redundancy for faster hot data downloads in a simplex coded storage.
In fairness-first system, each arriving request, asking for hot or cold data, is primarily assigned to its systematic server, e.g., requests for $a$ go to server-$a$, requests for $b$ go to server-$b$ in a $[a, b, a+b]$-system.
Redundant copies are launched to mitigate server side variability only for hot data requests, but in a restricted manner as follows. As soon as a hot data request moves into service at its systematic server, its recovery groups are checked to see if they are idle. A redundant copy of it is launched only at the idle recovery groups.
For instance in a $[a, b, a+b]$-system, when a request for hot data $a$ reaches the head of the line at server-$a$, a copy of it will be launched at server-$b$ and server-$a+b$ if both are idle, otherwise, it will be served only at server-$a$.

Fairness-first system aims to achieve perfect fairness towards cold data download. When a cold data request finds its systematic server busy with a redundant hot data request copy, the redundant copy in service will be immediately removed from the system. Thus, hot data requests can exploit redundancy for faster download only opportunistically, with zero effect on cold data download.

% Download time
We assume the request arrival stream for each stored data symbol is an independent Poisson process. The rate of request arrivals for each cold data symbol is set to be the same and denoted by $\lambda_c$, while the rate of request arrivals for hot data is denoted by $\lambda$. In fairness-first system, redundant copies of hot data requests have zero effect on the cold data download, hence the download traffic for each cold data symbol implements an independent $M/G/1$. Cold data download time can then be analyzed using the standard results on $M/G/1$ queues.

Download traffic for hot data implements a complex queueing system with redundancy. While a hot data request moves into service at its systematic server, it gets replicated to all of its idle recovery groups. A request copy assigned to a recovery group needs to wait for downloading from both servers. In addition, if there exists a cold data server in a recovery group, an assigned hot data request copy will get cancelled as soon as the cold data server receives a request. These interruptions modify the service time distribution of hot data requests on the fly.
The main difficulty in analyzing the hot data download time is that service times of subsequent hot data requests are not independent. The number of recovery groups that are found idle by a hot data request or interruption of its redundant copies by cold data arrivals tells us some information about the service time distribution of the subsequent hot data request.
% its service time will be distributed as $S_i \sim \min\{V_0, V^1_{2:1}, \dots, V^i_{2:1}\}$, where $V$'s denote the service times at the servers.

An upper and lower bound on the hot data download time is easy to find. As a lower bound we can use the split-merge model that is used to upper bound download time in replicate-to-all system, and a lower bound can be found along the same lines of Thm.~\ref{thm_reptoall_t_lb}.
\begin{theorem}
  Hot data download time in fairness-first system is bounded above by an $M/G/1$ queue with service time distribution $S_{t-log_2(t+1)}$ and is bounded below by an $M/G/1$ queue with service time distribution $S_t$, where
  \begin{equation}
    Pr\{S_i > s\} = Pr\{V > s\}\bigl(1 - Pr\{V \leq s\}^2\bigr)^i.
  \label{eq:eq_fairnessfirst_typei_servtime}
  \end{equation}
  $V$ denotes the service time at the servers.
\label{thm_fairnessfirst_ul_bound}
\end{theorem}
\begin{proof}
  The worst case scenario gives the upper bound, in which no hot data request can be replicated to any one of the recovery groups that include a cold data server, i.e., cold data servers are assumed to be busy all the time; $\lambda_c = 1/E[V]$. Remaining number of recovery groups with no cold data server is $t-log_2(t+1)$ (see Sec.~\ref{sec:sec_sys_model} for simplex code properties).
  The best case scenario gives the lower bound, in which each hot data request is fully replicated to all the $t$ recovery groups, i.e., cold data servers are assumed to be idle all the time; $\lambda_c = 0$.
  In either case, hot data download implements an $M/G/1$ queue with a service time distribution given in \eqref{eq:eq_fairnessfirst_typei_servtime}.
\end{proof}

Consider a low traffic regime under which the hot data requests do not wait in queue and go straight into service, i.e., each hot data request completes service before the next one arrives. We refer to this set of conditions at which the system operates as the \textit{low traffic assumption}. Under this assumption, each arriving request independently sees time averages (by PASTA), hence the service times of the subsequent hot data requests become i.i.d. Therefore, hot data download implements an $M/G/1$ queue under the low traffic assumption as stated in the following.
\begin{proposition}
  Let us denote the service times at the servers with $V$ and let the arrivals for each cold data follow a Poisson process of rate $\lambda_c$.
  Under the low traffic assumption, hot data download in fairness-first system implements an $M/G/1$ queue with a service time distribution given as
  \begin{equation}
  \begin{split}
    Pr\{S > s\} = C(s) \bigl[\alpha(s)(1-\rho) + \rho \bigr]^m,
  \end{split}
  \end{equation}
  where
  \begin{equation*}
  \begin{split}
    & m = log_2(t+1), \quad \rho = \lambda_c E[V], \\
    & \alpha(s) = 1 - Pr\{V \leq s\} \int_0^s e^{-\lambda_c v} Pr\{V = v\}dv, \\
    & C(s) = Pr\{V > s\}\bigl(1 - Pr\{V \leq s\}^2\bigr)^{t-m}.
  \end{split}
  \end{equation*}
%   \[ Pr\{S_i > s\} = Pr\{V > s\}\bigl(Pr\{V_{2:1} > s\}Pr\{V_{2:1} < X\} + Pr\{V_{2:1} > X\}\bigr)^i \]
%   where for $m = log_2(t+1)$
  
%   \begin{equation*}
%     E[V^j] = \sum_{i=0}^m \binom{m}{i} \left(1 - \frac{\lambda_c}{\mu}\right)^i \left(\frac{\lambda_c}{\mu}\right)^{m-i} E[W_i^j].
%   \end{equation*}
%   for $j \geq 1$. Given $V_i$ in eq.~\eqref{eq:eq_fairnessfirst_typei_servtime}, computation of $E[W_i^j]$ follows from the recursive procedure
%   \begin{equation*}
%   \begin{split}
%     E[W_i^j] &= Pr\{V_{t-m+i} \leq Exp(i\lambda_c)\}E[V_{t-m+i}^j] \\
%     &\quad + Pr\{V_{t-m+i} > Exp(i\lambda_c)\}E[W_{i-1}^j].
%   \end{split}
%   \end{equation*}
%   where the base case is $W_0^j \sim V_{t-m}^j$.
\label{prop_fairnessfirst_lowtraff}
\end{proposition}
\begin{proof}
  Cold data downloads at each cold data server implement an independent $M/G/1$ queue. Under stability, fraction of the time that a cold data server is busy with a cold data request is then given as $\rho = \lambda_c E[V]$.
  The number of recovery groups with a cold data server is given as $m = log_2(t+1)$.
  
  Low traffic assumption basically means that request arrivals for hot data do not wait in queue in the hot data server. Since the inter-arrival times are assumed to be Markovian, we can use PASTA and state that each recovery group with a cold data server is independently found idle with probability $1-\rho$ by each arriving hot data request. The number of such idle recovery groups seen by a hot data arrival is then distributed as $R \sim \textit{Binomial}(m, 1-\rho)$.
  
  Given $R = r$, a request will be simultaneously replicated to the systematic hot data server, and to the $t-m$ recovery groups without a cold data server, and to the $r$ recovery groups with a cold data server. Let us denote the service time distribution of such a request as $S_r$.
  Service time of the copy at the systematic server is distributed as $V$. At the recovery groups without a cold data server, we simply wait for downloading from two servers, hence the service time is distributed as $V_{2:2}$. Service time $V^\prime$ at the recovery groups with a cold data server requires a bit more attention since the hot data request copies get removed from service as soon as the cold data server receives a request. Let $X \sim Exp(\lambda_c)$ and denote the service time at the recovery server that stores a coded data (e.g., $a+b$) as $V$ and at the recovery server that stores a cold data (e.g., $b$) as $V_c$. Then, we can write the distribution of $V^\prime$ as
  \begin{equation*}
  \begin{split}
    Pr\{ V^\prime \leq s\} &= Pr\{V^\prime \leq s, V_c \leq X\} + Pr\{V^\prime \leq s, V_c > X\} \\
    & \stackrel{(a)}{=} Pr\{\max\set{V, V_c} \leq s, V_c \leq X\} \\
    &= Pr\{V \leq s\}Pr\{V_c \leq s, V_c \leq X\} \\
    &= Pr\{V \leq s\} \int_0^s e^{-\lambda_c v} Pr\{V_c = v\} dv.
  \end{split}
  \end{equation*}
  where $(a)$ is due to the cancellation of the copy at the cold data server due to an arrival of a cold data request, in other words
  \[ Pr\{V^\prime \leq s, V_c > X\} = 0. \]
  
  Now we are ready to write the distribution of $S_r$ as
  \begin{equation*}
  \begin{split}
    Pr\{ & S_r > s\} = Pr\{\min\{V_0, \{V_{2:2}\}_{i=1}^{t-m}, \{V^{\prime}\}_{i=1}^r\} > v\} \\
    &= Pr\{V > s\}\bigl(1 - Pr\{V \leq s\}^2\bigr)^{t-m} Pr\{V^{\prime} > s\}^r \\
    &= C(s) Pr\{V^{\prime} > s\}^r.
  \end{split}
  \end{equation*}
  
  Since each arriving hot data request independently samples from $R$, hot data downloads implement an $M/G/1$ queue with service time $S$, which is distributed as
  \begin{equation*}
  \begin{split}
    & Pr\{S > s\} = E_R[Pr\{S_R > s\}] \\
    &= \sum_{r=0}^m \binom{m}{r} (1-\rho)^r \rho^{m-r} Pr\{S_r > s\} \\
    &= C(s) \sum_{r=0}^m \binom{m}{r} \bigl[Pr\{V^{\prime} > s\}(1-\rho)\bigr]^r \rho^{m-r} \\
    &= C(s) \bigl[\alpha(s)(1-\rho) + \rho \bigr]^m.
  \end{split}
  \end{equation*}
\end{proof}

The approximation given in Prop.~\ref{prop_fairnessfirst_lowtraff} is exact when the hot data arrival rate is low. Fig.~\ref{fig:plot_simplex_ff_sim_vs_model} gives a comparison of the simulated average hot data download time and the approximation. Approximation seems to be fairly accurate as well when the hot data requests arrive at relatively higher rates. Accuracy of the approximation diminishes with increasing arrival rate. This is because the length of the busy periods in the hot data server increases which makes the independence assumption on the request service times invalid, hence the adopted $M/G/1$ model starts to deviate from the actual behavior.
\begin{figure}[htbp]
  \centering
  \begin{subfigure}[h]{.4\textwidth}
    \centering
    \includegraphics[width=1\textwidth, keepaspectratio=true]{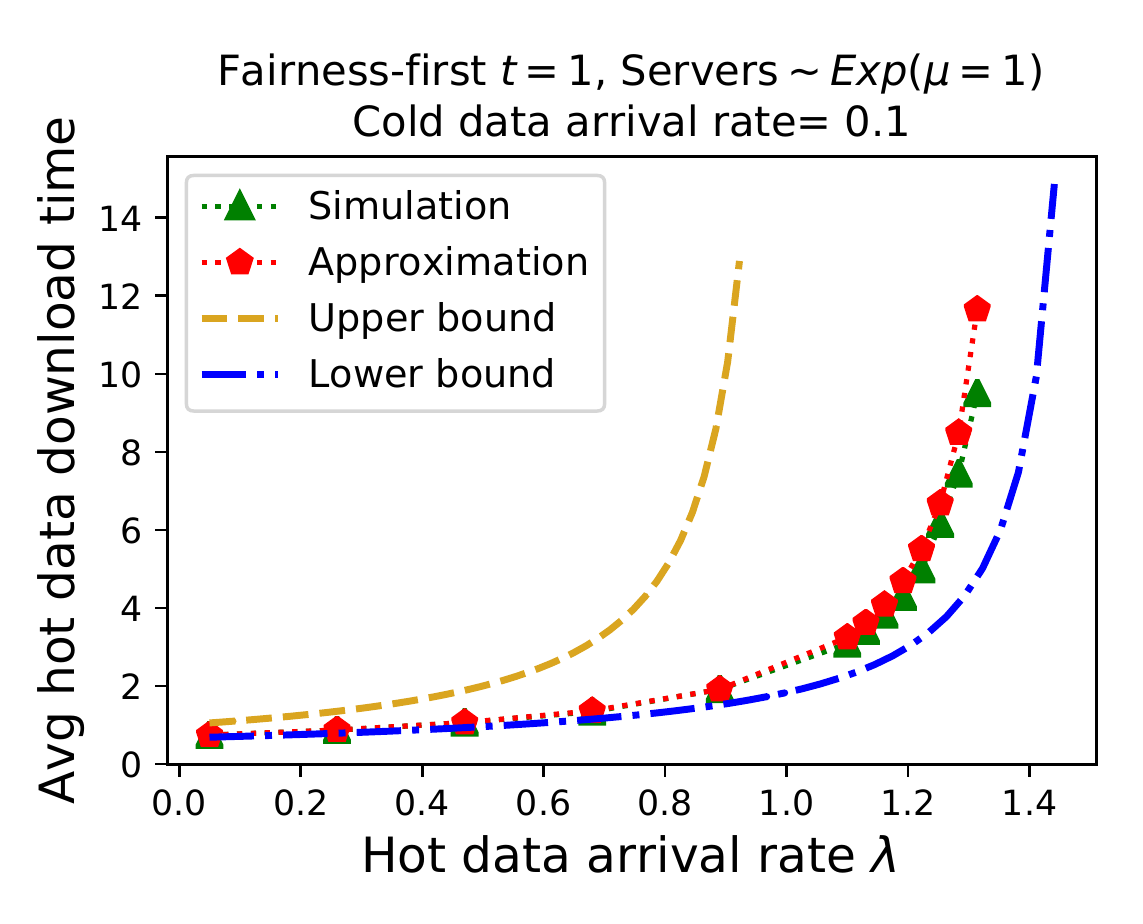}
  \end{subfigure}
  \begin{subfigure}[h]{.4\textwidth}
    \centering
    \includegraphics[width=1\textwidth, keepaspectratio=true]{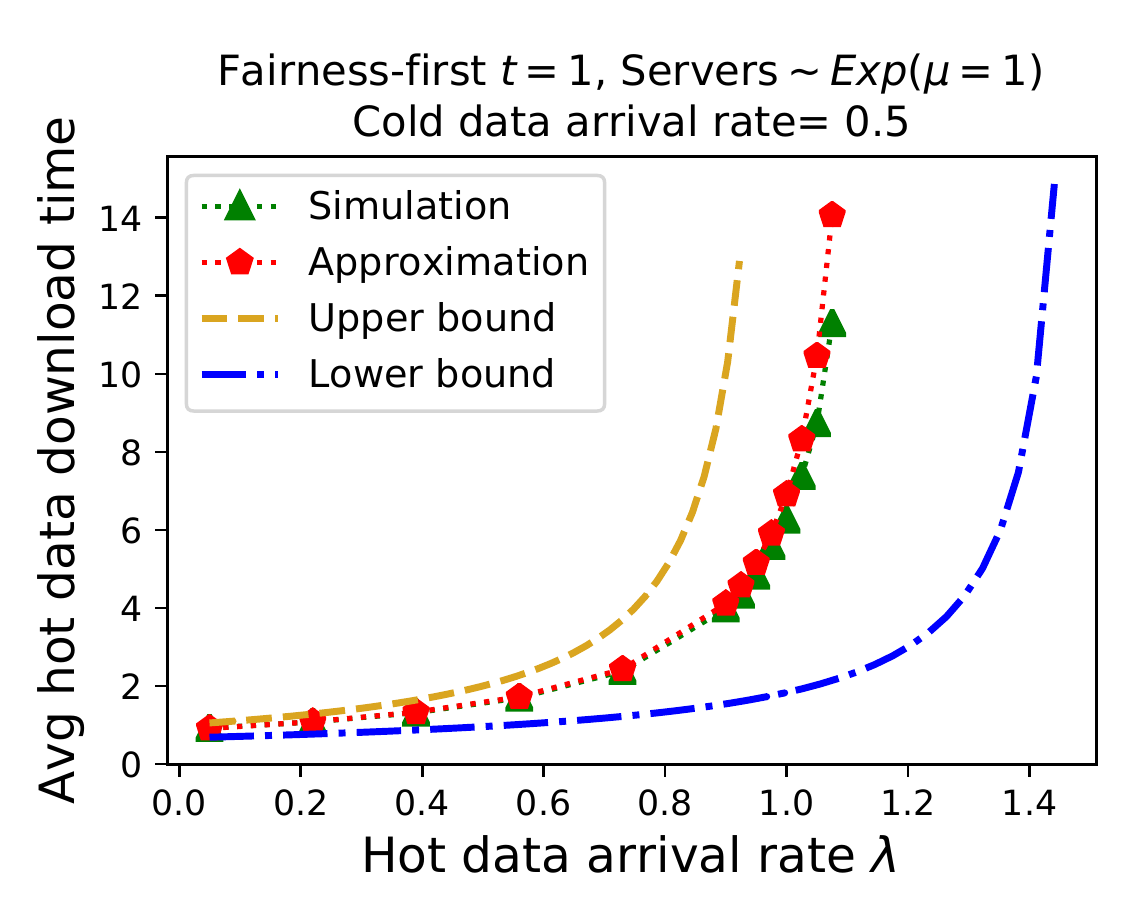}
  \end{subfigure}
  \begin{subfigure}[h]{.4\textwidth}
    \centering
    \includegraphics[width=1\textwidth, keepaspectratio=true]{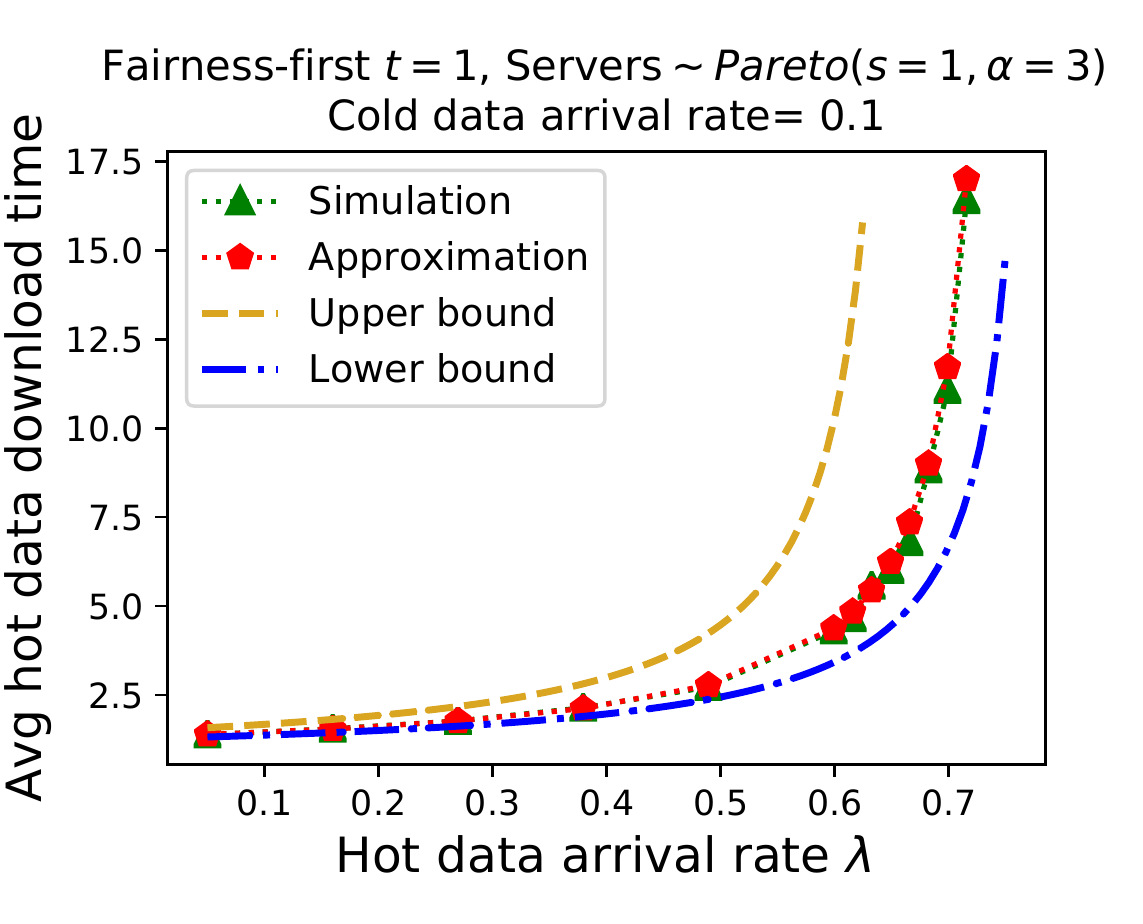}
  \end{subfigure}
  \begin{subfigure}[h]{.4\textwidth}
    \centering
    \includegraphics[width=1\textwidth, keepaspectratio=true]{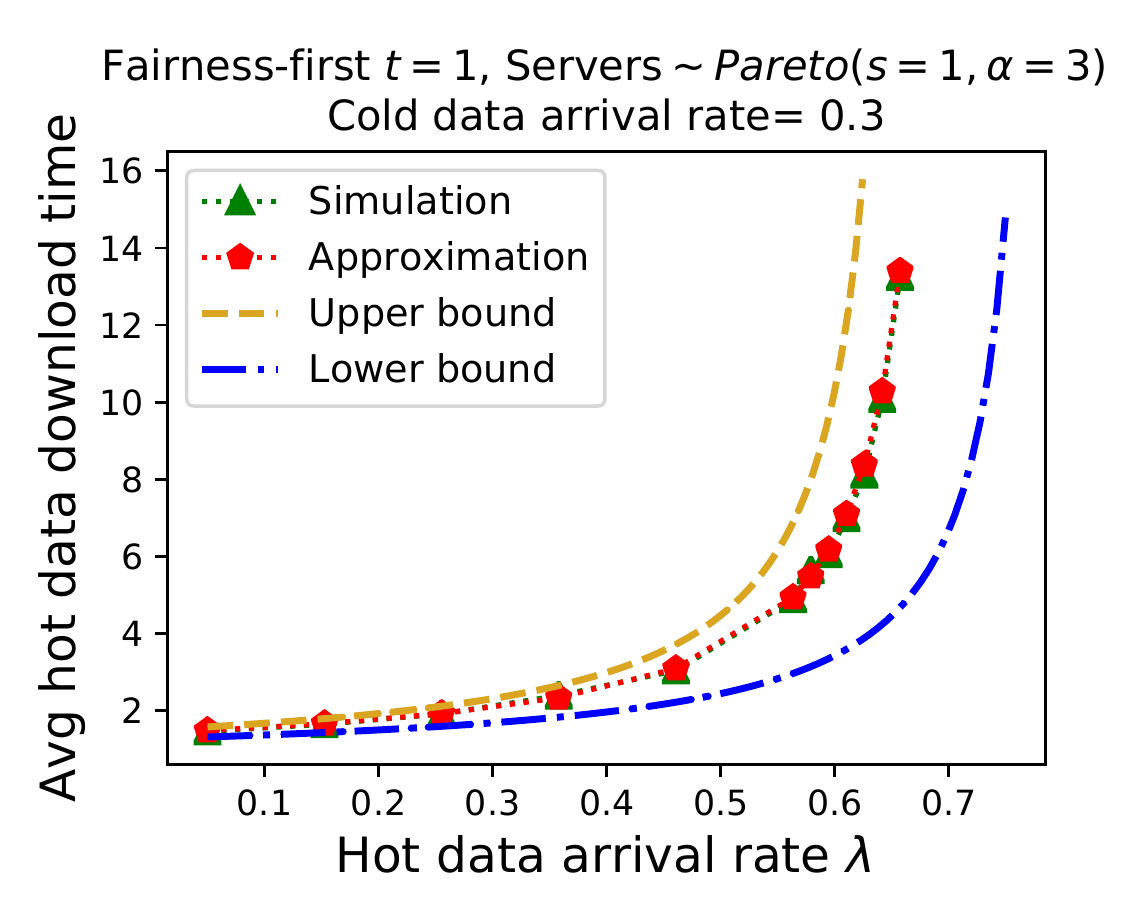}
  \end{subfigure}
  \caption{Comparison of the upper and lower bounds in Thm.~\ref{thm_fairnessfirst_ul_bound}, the $M/G/1$ approximation in Prop.~\ref{prop_fairnessfirst_lowtraff}, and the simulated average hot data download time in fairness-first system.}
  \label{fig:plot_simplex_ff_sim_vs_model}
\end{figure}

% Next we discuss the pain and gain of using replicate-to-all over fairness-first scheduling.
To recap, replicate-to-all scheduler aims to exploit download with redundant copies for each request arrival with no distinction. On the other hand, fairness-first scheduler allows for download with redundancy only for hot data requests by eliminating any negative impact caused on the cold data download by the redundant hot data request copies.
At low arrival rate, one would expect replicate-to-all to outperform fairness-first in both hot or cold data download since the queues do not build up much and the redundant request copies do not significantly increase the waiting times of the requests. However at higher arrival regimes, replicate-to-all overburdens the cold data servers with redundant requests and cause great pain in the waiting times experienced by the cold data requests. Fig.~\ref{fig:plot_reptoall_over_ff} gives an illustration of these observations and compares the gain of replicate-to-all over the fairness-first scheduler. After certain threshold in hot data arrival rate, replicate-to-all starts incurring more pain in percentage (negative gain) on cold data download than the gain achieved for hot data download. This is intuitive since download with redundancy has diminishing returns as the arrival rate gets higher, and the high waiting times caused by the excessive number of redundant request copies in the system has a much greater impact on the cold data download time.
\begin{figure}[htbp]
  \centering
  \begin{subfigure}[h]{.4\textwidth}
    \centering
    \includegraphics[width=1\textwidth, keepaspectratio=true]{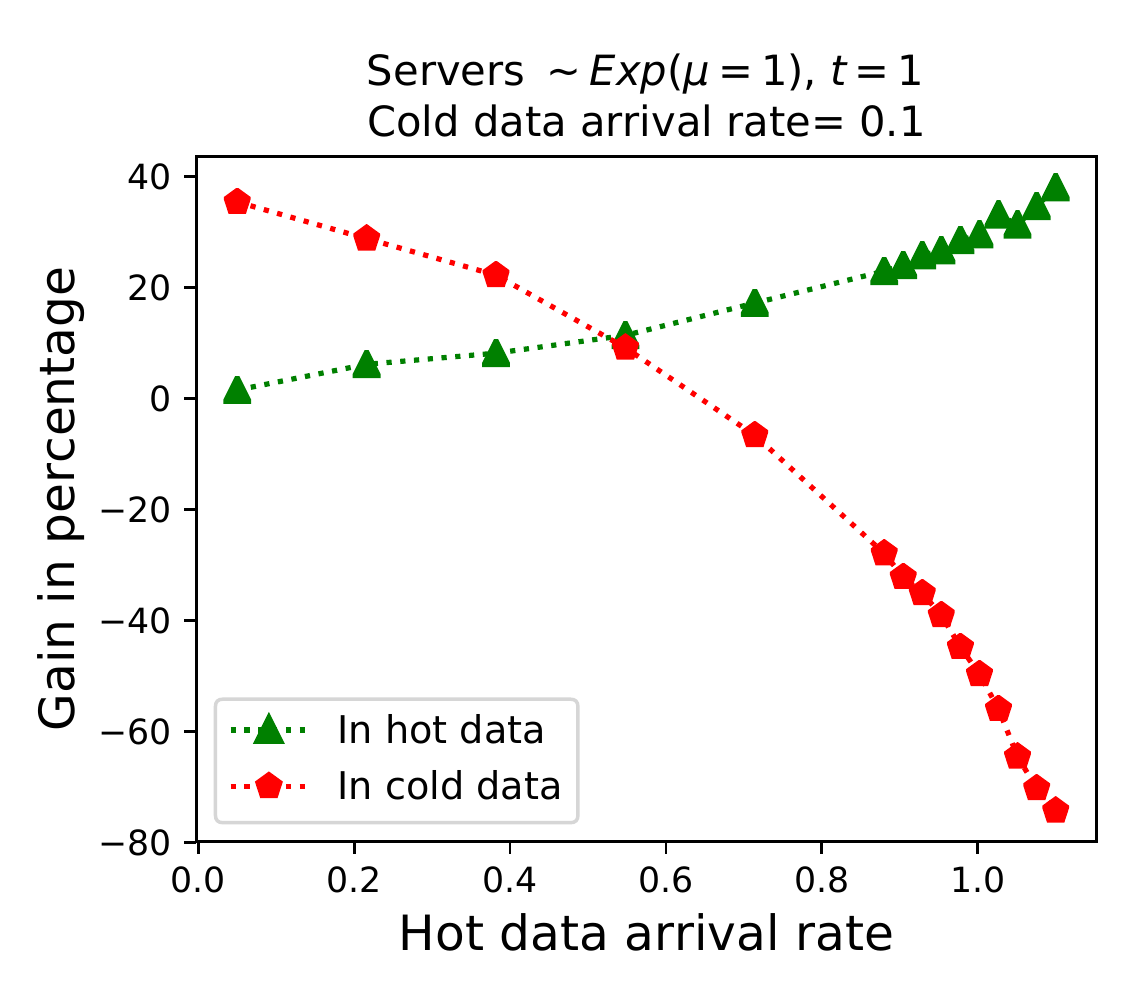}
  \end{subfigure}
  \begin{subfigure}[h]{.4\textwidth}
    \centering
    \includegraphics[width=1\textwidth, keepaspectratio=true]{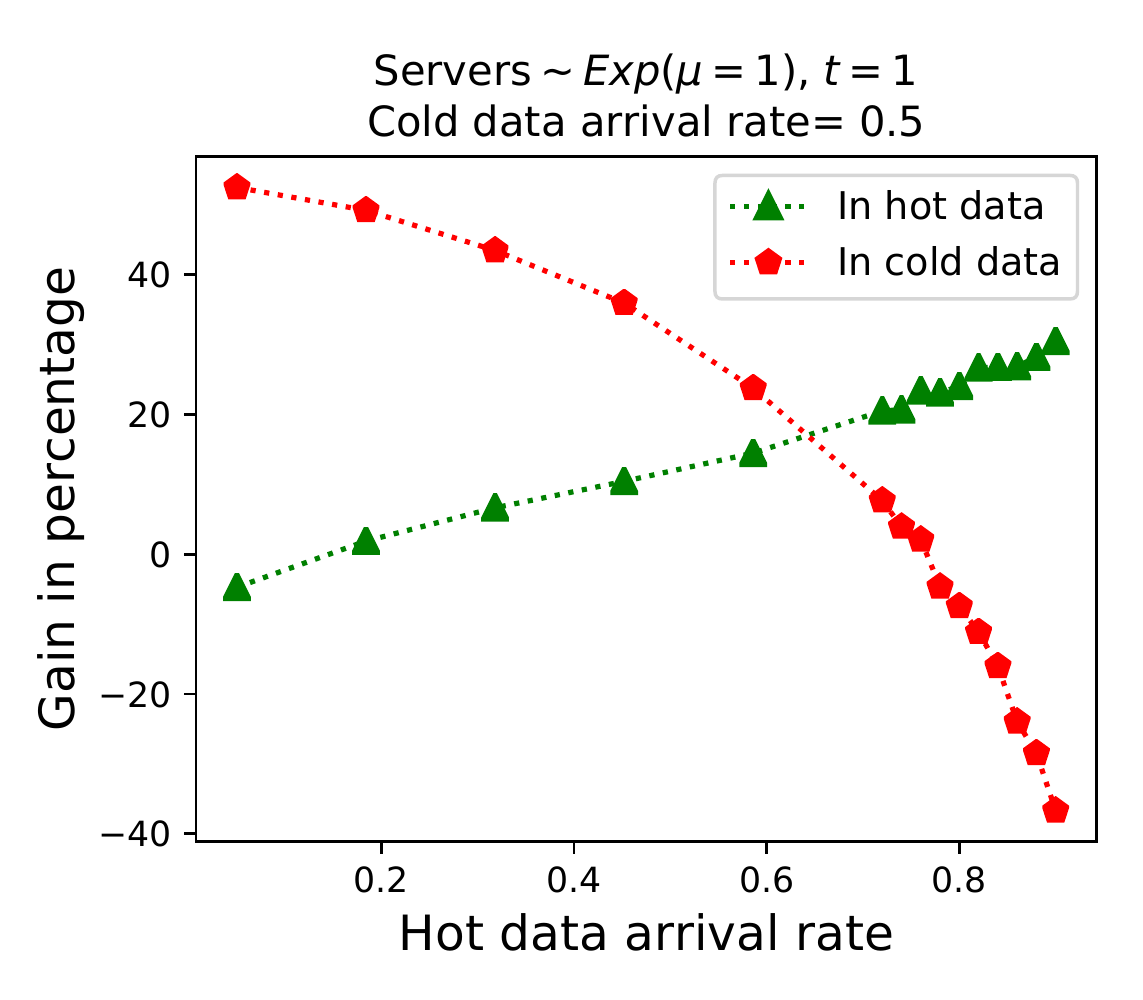}
  \end{subfigure}
  \caption{Gain or pain of replicate-to-all over fairness-first scheduler in average hot and cold data download time. Negative gain means pain, which indicates that replicate-to-all performs worse than fairness-first in terms of the hot or cold data download time. }
  \label{fig:plot_reptoall_over_ff}
\end{figure}

% #########################################  Acknowledgments  ######################################## %
\section*{Acknowledgments}
This material is based upon work supported by the National Science Foundation under Grant No.~CIF-1717314.

% %%%%%%%%%%%%%%%%%%%%%%%%%%%%%%%%%%%%%%%%%%%%%%%%%%%%%%%%%%%%%%%%%%%%%%%%%%%%%%%
\bibliographystyle{unsrt}
\bibliography{references}

% #########################################  Appendix  ######################################## %
\appendix
\section{Appendix}
% -------------------------  Conjecture for Simplex(t)  ------------------------- %
\subsection{On the Conj.~\ref{conj_reptoall_t_fjs} in Sec.\ref{sec:sec_reptoall}}
\label{subsec:subsec_simplex_t_conjecture}
% Is Theorem \ref{thm_necc_cond_for_conj_reptoall_t_fjs} a necessary condition for Conjecture \ref{conj_reptoall_t_fjs}?
% \mehmet{Not really, implemented a simple case where the result above does not imply this; trials fail to hold this inequality for high values of $i$. For now we will keep it as a conjecture.}
We here present a Theorem that helps to build an intuition for Conj.~\ref{conj_reptoall_t_fjs}. We firstly assume that the service times at the servers are exponentially distributed with rate $\mu$. Let us redefine the state of the replicate-to-all system $S$ as the type (see Lm.~\ref{lm_reqservtime_dists} for details) of service start for the request at the head of the system, so $S \in \{0, \ldots, t\}$. System state transitions correspond to time epochs at which a request moves into service (see Def.~\ref{def_req_servstart} for the definition of request service start times).

Given that Conj.~\ref{conj_reptoall_t_fjs} holds i.e., $f_j > f_{j+1}$ for $j = 0, \dots, t-1$, one would expect the average drift at state-$j$ to be towards the lower rank state-$i$'s with $i < j$.
In the Theorem below we prove this indeed holds. However, it poses neither a necessary nor a sufficient condition for Conj.~\ref{conj_reptoall_t_fjs}. The biggest in proving the conjecture is that a possible transition exists between every pair of states. Finding transition probabilities requires the exact analysis of the system, which proved to be formidable. For instance, if the transitions were possible only between the consecutive states $j$ and $j+1$ as in a birth-death chain, Theorem below would have implied the conjecture.
\begin{theorem}
  In replicate-to-all system, let $J_i$ be the type of service time distribution for request-$i$. Then $Pr\{J_{i+1} > j|J_i=j\} < 0.5$. In other words, given that a request with type-$j$ service time distribution, the subsequent request is more likely to be served with type-$i$ distribution such that $i < j$.
  \label{thm_necc_cond_for_conj_reptoall_t_fjs}
\end{theorem}
\begin{proof}
Let us define $L_k(t)$ as the absolute difference of queue lengths at repair group $k$ and time $t$. Suppose that the $i$th request moves in service at time $\tau$ and its service time is type-$j$; namely $J_i=j$. Let also $A$ denote the event that $L_k(\tau) > 1$ for every repair group $k$ that has a leading server at time $\tau$, we refer to other recovery groups as non-leading. Then, the following inequality holds
\begin{equation}
  Pr\{J_{i+1} > j|J_i=j, A\} > Pr\{J_{i+1} > j|J_i=j\}
\label{eq:eq_Jip1_g_j_given_Ji_eq_j__ineq}
\end{equation}
Event $A$ guarantees that $J_{i+1} \geq j$ i.e., $Pr\{J_{i+1} \geq j|J_i=j, A\}=1$, because even when none of the leading servers advances before the $i$th request departs, the next request will make at least type-$j$ start.

We next try to find
\[ Pr\{J_{i+1} > j|J_i=j, A\} = 1 - Pr\{J_{i+1}=j|J_i=j, A\}. \]
Suppose that the $i$th request departs at time $\tau^\prime$ and without loss of generality, let the repair group $k$ has a leading server only if $k \leq j$. The event $\{J_{i+1}=j|J_i=j, A\}$ is equivalent to the event
\begin{equation}
  B_j = \{L_k(\zeta) < 2; \zeta \in [\tau, \tau^\prime], ~j < k \leq t\} \text{ for } 0 < j < t-1.
\label{eq:eq_event_Bj}
\end{equation}
This is because for the $(i+1)$th request to have type-$(j+1)$ service time, one of the servers in at least one of the non-leading recovery groups should advance by \textit{at least two} replicas before the $i$th request terminates.

Event $B_j$ can be re-expressed as
\[ B_j = \bigcup\limits_{l=0}^{t-j} C_l; \quad C_l = \{L_{k_i}(\tau)=1; 1 \leq i \leq l, j < k_i \leq t\}. \]
Event $C_l$ describes that $l$ non-leading recovery groups start leading by one replica before the departure of $i$th request. Given that there are currently $j$ leading recovery groups, let $p_j^{+1}$ be the probability that a non-leading repair group starts to lead by one before the departure of $ith$ request, and $p_j^D$ be the probability that $i$th request departs first. Then, we can write
\[ Pr\{C_l\} = p_{j+l}^D\prod\limits_{k=j}^{j+l-1}p_k^{+1}. \]
Since the events $C_l$ for $l = 0, \dots, t-j$ are disjoint, $Pr\{B_j\} = \sum\limits_{l=0}^{t-j} Pr\{C_l\}$ and we get the recurrence relation
\[ Pr\{B_j\} = p_j^D + p_j^{+1}Pr\{B_{j+1}\} \]
Since the service times at the servers are exponentially distributed, the probabilities are easy to find as
\[ p_j^D = \frac{1+j}{1+2t}, \quad p_j^{+1} = \frac{2(t-j)}{1+2t}. \]
% $p_j^D = (\gamma+j\mu)/(\gamma+2t\mu)$ and $p_j^{+1} = (2(t-j)\mu)/(\gamma+2t\mu)$ where $\gamma$ and $\mu$ are respectively service rates of the systematic server and repair servers.
Then we find
\begin{equation*}
  \begin{split}
    Pr\{B_{t-1}\} &= p_{t-1}^D + p_{t-1}^{+1}p_t^D \\
    &= \frac{1+(t-1)}{1+2t} + \frac{2}{1+2t}\frac{1+t}{1+2t} \\
    &= \frac{1+t}{1+2t} + \frac{1}{(1+2t)^2} \\
    &= 1 - \frac{t}{1+2t} + \frac{1}{(1+2t)^2} > \frac{1}{2}
  \end{split}
\end{equation*}
Suppose $Pr\{B_{j+1}\} > 0.5$, then we have
\begin{equation*}
  \begin{split}
    Pr\{B_j\} &= \frac{1+j}{1+2t} + \frac{2(t-i)}{1+2t}Pr\{B_{j+1}\} \\
    &> \frac{1+i}{1+2t} + \frac{(t-i)}{1+2t}\\
    &= \frac{1+t}{1+2t} = \frac{1}{2} + \frac{1/2}{1+2t} > \frac{1}{2}
  \end{split}
\end{equation*}
Knowing that $Pr\{B_{t-1}\} > 0.5$ together with $Pr\{B_k\} > 0.5$, and given that $Pr\{B_{k+1}\} > 0.5$ we have $Pr\{B_j\} > 0.5$ for each $j$.
% \[ Pr\{J_{i+1} > j|J_i=j\} < Pr\{J_{i+1} > j|J_i=j, A\}\]

By \eqref{eq:eq_Jip1_g_j_given_Ji_eq_j__ineq} and \eqref{eq:eq_event_Bj}, we find
\begin{equation*}
\begin{split}
  Pr\{J_{i+1} > j|J_i=j\} &< Pr\{J_{i+1} > j|J_i=j, A\} \\
  &= 1 - Pr\{B_j\} < 0.5.
\end{split}
\end{equation*}
which tells us that for any request with type-$j$ service time, the subsequent request is more likely to have service time with type less than $j$.
\end{proof}

% -------------------------  Guessing-based analysis of MP for Simplex(t=1)  ------------------------- %
\subsection{Approximate analysis of the Markov process for the replicate-to-all system with availability one}
\label{subsec:subsec_reptoall_t1_pyramid_mp_analysis}
Here we give an approximate analysis of the process illustrated in Fig.~\ref{fig:fig_reptoall_t1_mp__high_traff_approx} with a guessing based local balance equations approach. Consider the case where all three servers operate at the same rate, i.e., $\alpha = \beta = \mu$, which makes the pyramid process symmetric around the center column, i.e., $p_{k,(i,0)} = p_{k,(0,i)}$ for $1 \leq i \leq k$. Thus, the following discussion is given in terms of the states on the right side of the pyramid.

Observe that under low traffic load, system spends almost the entire time in states $(0,(0,0))$, $(1,(0,0))$, $(1,(0,1))$ and $(1,(1,0))$. Given this observation, notice that the rate entering into $(1,(0,0))$ due to request arrivals is equal to the rate leaving the state due to request departures. To help with guessing the steady-state probabilities, we start with the assumption that rate entering into a state due to request arrivals is equal to the rate leaving the state due to request departures. This gives us the following relation between the steady-state probabilities of the column-wise subsequent states:
\begin{equation}
  p_{k,(i,0)} = \frac{\lambda}{\gamma+2\mu}p_{k-1,(i,0)},~0 \leq i \leq k.
  \label{eq:eq_geometric_over_column}
\end{equation}
Let us define $\tau = \lambda/(\gamma+2\mu)$. This relation allows us to write $p_{k,(i,0)} = \tau^{k-i}p_{i,(i,0)}$. However this obviously won't hold for higher arrival rates since at higher arrival rates some requests wait in the queue, which requires the rate entering into a state due to request arrivals to be higher than the rate leaving the state due to task completions.
%  Therefore, we need to keep in mind that parameter $\tau$ is an increasing function of $\lambda$.
To be used in the following discussion, first we write $p_{1,(1,0)}$ in terms of $p_{0,(0,0)}$ from the global balance equations as the following.
\begin{equation}
  \label{eq:eq_p_0_global_balance}
  \begin{split}
    & \lambda p_{0,(0,0)} = \gamma p_{1,(0,0)} + 2(\gamma+\mu)p_{1,(1,0)}, \\
    & p_{1,(1,0)} = \frac{\lambda-\gamma\tau}{2(\gamma+\mu)}p_{0,(0,0)}
  \end{split}
\end{equation}
For the nodes at the far right side of the pyramid, we can write the global balance equations and solve the corresponding recurrence relation as the following.
\begin{equation}
\begin{split}
  & p_{i,(i,0)}(\lambda+\mu+\gamma) = p_{i,(i-1,0)}\mu + p_{i+1,(i+1,0)}(\mu+\gamma),~i \geq 1, \\
  & p_{i+2,(i+2,0)} = b p_{i+1,(i+1,0)} + a p_{i,(i,0)},~i \geq 0 \text{ where} \\
  &\qquad b=1+\frac{\lambda}{\mu+\gamma}, \qquad a=\frac{-\tau\mu}{\gamma+\mu}, \\
  &\implies p_{i,(i,0)} = \frac{A}{r_0^i} + \frac{B}{r_1^i} \text{ where } \\
  &\qquad B = \frac{r_0p_{0,(0,0)} + (p_{1,(1,0)}-bp_{0,(0,0)})r_0r_1}{r_0-r_1}, \\
  &\qquad A = p_{0,(0,0)}-B \text{ where } \\
  &\qquad (r_0, r_1) = \Bigl(\frac{-b-\sqrt{\Delta}}{2a}, \frac{-b+\sqrt{\Delta}}{2a}\Bigr);~\Delta=b^2+4a, \\
  & p_{k,(i,0)} = p_{k,(0,i)} = \tau^{k-i}\Bigl(\frac{A}{r_0^i} + \frac{B}{r_1^i}\Bigr),~0 \leq i \leq k.~\label{eq:eq_edge_global_balance}
\end{split}
\end{equation}

Even though the algebra does not permit much cancellation, one can find the unknowns $A$ and $B$ above by computing $p_{0,(0,0)}$ as follows.
\begin{longaligned}
  & \sum_{k=0}^{\infty} p_{k,(0,0)} + \sum_{i=1}^{\infty} \sum_{k=i}^{\infty} (p_{k,(i,0)} + p_{k,(0,i)}) \\
  &= \frac{p_{0,(0,0)}}{1-\tau} + \frac{2}{1-\tau}\sum_{i=1}^{\infty}p_{i,(i,0)} \quad (\tau < 1) \\
  &= \frac{p_{0,(0,0)}}{1-\tau} + \frac{2}{1-\tau}\sum_{i=1}^{\infty}\left(\frac{A}{r_0^i} + \frac{B}{r_1^i}\right) \\
  &= \frac{p_{0,(0,0)}}{1-\tau} + \frac{2}{1-\tau}\left(\frac{A}{r_0-1} + \frac{B}{r_1-1}\right) \quad ({\eqref{eq:eq_edge_global_balance}},r_0,r_1 > 1) \\
  &= \frac{p_{0,(0,0)}}{1-\tau} + \frac{2}{1-\tau}\left(\frac{(p_{0,(0,0)}-B)(r_1-1) + B(r_0-1)}{(r_1-1)(r_0-1)}\right) \\
  &= \frac{p_{0,(0,0)}}{1-\tau} + \frac{2}{1-\tau}\left(\frac{B(r_0-r_1)+p_{0,(0,0)}(r_1-1)}{(r_1-1)(r_0-1)}\right) \\
  &= \frac{p_{0,(0,0)}}{1-\tau} + \\
  & \quad \frac{2}{1-\tau}\left(\frac{(r_0p_0+r_0r_1(p_{1,(1,0)}-bp_{0,(0,0)}))+p_{0,(0,0)}(r_1-1)}{(r_1-1)(r_0-1)}\right) \\
  &= p_{0,(0,0)}\left[\frac{1 + 2\Bigl(\frac{r_0 + r_0 r_1\bigl((\lambda-\gamma\tau)/(2\gamma + 2\mu) - b\bigr) + r_1 - 1}{(r_1-1)(r_0-1)} \Bigr)}{1-\tau}\right] = 1.
%   & \implies p_{0,(0,0)} = \frac{1 - \tau}{1 + 2\Bigl(\frac{r_0 + r_0 r_1\bigl((\lambda-\gamma\tau)/(2\gamma + 2\mu) - b\bigr) + r_1 - 1}{(r_1-1)(r_0-1)} \Bigr)}.
\end{longaligned}
Simulation results show that the model for $p_{k,(i,0)}$ discussed above is proper in structure i.e., $p_{k,(i,0)}$ decreases exponentially in $k$ and $i$. However, simulations show that $\tau(\lambda) = k(\gamma,\mu)\lambda/(\gamma+2\mu)$. For instance, for $\gamma=\mu$ we can find that $k(\gamma,\mu) \simeq 0.3$. Nevertheless this does not permit to find a general expression for $k(\gamma,\mu)$.

% -------------------------  Service rate allocation in Simplex(t=1)  ------------------------- %
\subsection{Service rate allocation in the replicate-to-all system with availability one}
\label{subsec:subsec_dE_T_dro_less_than_zero}
Here we present the algebra that shows $\partial E[T_{approx}]/\partial\rho < 0$ where $E[T_{approx}]$ is given in Thm.~\ref{thm_reptoall_t1_hightraff_approx}. Define $C = \gamma+2\mu$ and $\rho = \gamma/\mu$, then the followings can be calculated

\begin{longaligned}
  & \hat{f}_0 = \frac{1}{1 + 2/[\rho(\rho+2)]}, \quad\quad \frac{\partial \hat{f}_0}{\partial \rho} = \frac{4(\rho+1)}{(\rho^2+2\rho+2)^2}, \\
  & E[V_1] = \frac{\rho+2}{C(\rho+1)}, \quad\quad E[V_1^2]=\frac{2}{C^2}\Bigl(\frac{\rho+2}{\rho+1}\Bigr)^2, \\
  & E[V_0] = \frac{2(\rho+2)}{C(\rho+1)} - \frac{1}{C}, \quad\quad E[V_0^2]=\frac{4}{C^2}\Bigl(\frac{\rho+2}{\rho+1}\Bigr)^2 - \frac{2}{C^2}, \\
  & \frac{\partial E[V_1]}{\partial \rho} = \frac{-1}{C(\rho+1)^2}, \quad\quad \frac{\partial E[V_1^2]}{\partial \rho} = \frac{-4(\rho+2)}{C^2(\rho+1)^3}, \\
  & \frac{\partial E[V_0]}{\partial \rho} = \frac{-2}{C(\rho+1)^2}, \quad\quad \frac{\partial E[V_0^2]}{\partial \rho} = \frac{-8(\rho+2)}{C^2(\rho+1)^3}, \\
  % Serv time First moment
  & (i)\quad E[V_{lb}] = E[V_1] + \hat{f}_0(E[V_0] - E[V_1]), \\
  & \begin{aligned}
    \frac{\partial E[V_{lb}]}{\partial \rho} =& \frac{\partial E[V_1]}{\partial \rho} + \frac{\partial \hat{f}_0}{\partial \rho}(E[V_0]-E[V_1]) \\
    &+ \frac{\partial (E[V_0]-E[V_1])}{\partial \rho}\hat{f}_0 \\
    \end{aligned} \\
  & \begin{aligned}
    = &\frac{-1}{C(\rho+1)^2} + \frac{4(\rho+1)}{(\rho^2+2\rho+2)^2}\frac{1}{C(\rho+1)} \\
    &- \frac{1}{C(\rho+1)^2}\frac{\rho(\rho+2)}{\rho^2+2\rho+2} \\
    \end{aligned} \\
  &= \frac{1}{C}(\frac{-1}{(\rho+1)^2} + \frac{4}{(\rho^2+2\rho+2)^2} - \frac{\rho(\rho+2)}{(\rho+1)^2(\rho^2+2\rho+2)}) \\
  &= \frac{-(\rho^2+2\rho+2)^2 + 4(\rho+1)^2 - (\rho^2+2\rho)(\rho^2+2\rho+2)}{C(\rho+1)^2(\rho^2+2\rho+2)} \\
  &= \frac{-2(\rho+1)^2(\rho^2+2\rho+2) + 4(\rho+1)^2}{C(\rho+1)^2(\rho^2+2\rho+2)} \\
  &= \frac{-2(\rho+1)^2(\rho^2+2\rho)}{C(\rho+1)^2(\rho^2+2\rho+2)} < 0, \\
  % Serv time Second moment
  & (ii)\quad E[V_{lb}^2] = E[V_1^2] + \hat{f}_0(E[V_0^2] - E[V_1^2]), \\
  & \begin{aligned}
    \frac{\partial E[V_{lb}^2]}{\partial \rho} =& \frac{\partial E[V_1^2]}{\partial \rho} + \frac{\partial \hat{f}_0}{\partial \rho}(E[V_0^2]-E[V_1^2]) \\
    &+ \frac{\partial (E[V_0^2]-E[V_1^2])}{\partial \rho}\hat{f}_0 \\
    \end{aligned} \\
  & \begin{aligned}
    =& \frac{-4(\rho+2)}{C^2(\rho+1)^3} + \frac{8(\rho+1)}{C^2(\rho^2+2\rho+2)^2}((\frac{\rho+2}{\rho+1})^2-1) \\
    &- \frac{4\rho(\rho+2)^2}{C^2(\rho+1)^3(\rho^2+2\rho+2)} \\
   \end{aligned} \\
  & \begin{aligned}
   =& \frac{8}{C^2}(\frac{2\rho+3}{(\rho+1)(\rho^2+2\rho+2)^2} - \frac{\rho+2}{(\rho+1)^3}) \\
   &- \frac{4\rho(\rho+2)^2}{C^2(\rho+1)^3(\rho^2+2\rho+2)} < 0.
   \end{aligned} \\
  % Average download time
  & (iii)\quad E[T_{approx}] = E[V_{lb}] + \frac{\lambda E[V_{lb}^2]}{2(1 - \lambda E[V_{lb}])} \\
  & \begin{aligned}
    \frac{\partial E[T_{approx}]}{\partial \rho} =& \frac{\partial E[V_{lb}]}{\partial \rho} \\
    &+ \frac{\lambda}{2}\frac{(\frac{\partial E[V_{lb}^2]}{\partial \rho}(1-\lambda E[V_{lb}]) + \frac{\partial E[V_{lb}]}{\partial \rho}\lambda E[V_{lb}^2])}{(1-\lambda E[V_{lb}])^2} \\
    \end{aligned} \\
  &= \frac{\partial E[V_{lb}]}{\partial \rho}(1 + \frac{\lambda^2 E[V_{lb}^2]}{2(1-\lambda E[V_{lb}])}) + \frac{\partial E[V_{lb}^2]}{\partial \rho}\frac{\lambda}{2(1-\lambda E[V_{lb}])} \\
  & \text{under stability } \lambda E[V_{lb}] < 1 \text{ and we found above } \\
  & \frac{\partial E[V_{lb}]}{\partial \rho} < 0, \frac{\partial E[V_{lb}^2]}{\partial \rho} < 0 \text{ which shows that } \frac{\partial E[T_{approx}]}{\partial \rho} < 0.
\end{longaligned}

% -------------------------  Matrix Analytic for Simplex(t=1)  ------------------------- %
\subsection{Matrix Analytic solution for the replicate-to-all system with availability one}
\label{subsec:subsec_reptoall_t1_matrix_analytic}
Defining $\delta = \alpha+\beta+\gamma+\lambda$, the sub-matrices forming $\bm{Q}$ are given below.
\begin{equation*}
  \begin{split}
  \bm{F}_0 &= \left[\begin{array}{cccc}
-\lambda & 0 & \lambda & 0  \\
\alpha+\gamma & \beta-\delta & 0 & 0  \\
\gamma & \beta & -\delta & \alpha \\
\beta+\gamma & 0 & 0 &\alpha-\delta \\
\end{array}\right], \\
  \bm{H}_0 &= \left[\begin{array}{ccccc}
0 & 0 & 0 & 0 & 0 \\
0& \lambda & 0 & 0 & 0 \\
0 & 0 & \lambda & 0 & 0\\
0 & 0 &  0 & \lambda & 0
\end{array}\right], \\
  \bm{L}_0 &= \left[\begin{array}{cccc}
0 & \alpha+\gamma & 0 & 0 \\
0 & 0 & \alpha+\gamma & 0 \\
0 & 0 & \gamma & 0 \\
0 & 0 & \beta+\gamma & 0\\
0 & 0 & 0 & \beta+\gamma
\end{array}\right], \\
  \bm{F} &= \left[\begin{array}{ccccc}
\beta-\delta & 0 & 0 & 0 & 0 \\
\beta & -\delta & 0 & 0 & 0 \\
0 & \beta & -\delta & \alpha & 0\\
0 & 0 & 0 & -\delta & \alpha \\
0 & 0 &  0 & 0 & \alpha-\delta
\end{array}\right], \\
  \bm{L} &= \left[\begin{array}{ccccc}
0 & \alpha+\gamma & 0 & 0 & 0 \\
0 & 0 & \alpha+\gamma & 0 & 0 \\
0 & 0 & \gamma & 0 & 0 \\
0 & 0 & \beta+\gamma & 0 & 0\\
0 & 0 & 0 & \beta+\gamma & 0
\end{array}\right], \\
  \bm{H} &= \left[\begin{array}{ccccc}
\lambda & 0 & 0 & 0 & 0 \\
0 & \lambda & 0 & 0 & 0 \\
0 & 0 & \lambda & 0 & 0 \\
0 & 0 & 0 & \lambda & 0\\
0 & 0 & 0 & 0 & \lambda
\end{array}\right].
  \end{split}
\end{equation*}

\end{document}